\documentclass[11pt, reqno]{amsart}
\usepackage[foot]{amsaddr}
\usepackage{amsmath, amssymb}
\usepackage{amsthm}
\usepackage[
  margin=1in,
  includefoot
,  footskip=10pt,
]{geometry}

\usepackage[english]{babel}
\usepackage{hyperref}
\usepackage{multirow}
\usepackage{enumerate} 
\usepackage{caption}
\usepackage{subcaption}

\usepackage{accents}

\usepackage{mathdots}

\usepackage{tikz}
\usepackage{tikz-cd}
\usetikzlibrary{calc}
\usepackage{color, mathtools,epsfig, graphicx}
\usetikzlibrary{positioning}

\newtheorem{theorem}{Theorem}[section]

\newtheorem{proposition}[theorem]{Proposition}

\newtheorem{corollary}[theorem]{Corollary}

\newtheorem{lemma}[theorem]{Lemma}

\theoremstyle{definition}

\newtheorem{definition}[theorem]{Definition}

\theoremstyle{remark}

\addto\extrasenglish{

}

\numberwithin{equation}{section}

\newcommand{\C}{\mathbb{C}}
\newcommand{\R}{\mathbb{R}}

\newcommand{\Z}{\mathbb{Z}}

\newcommand{\pain}[1]{\operatorname{P}_{\mathrm{#1}}}
\newcommand{\defeq}{\vcentcolon=}

\begin{document}

\title[Delay Painlev\'e-I equation, associated polynomials and Masur-Veech volumes]{Delay Painlev\'e-I equation, associated polynomials and Masur-Veech volumes}

\author[J. Gibbons]{J. Gibbons$^{1}$}

\address{$^1$Department of Mathematics, Imperial College London, London SW7 2AZ, UK.}
\email{gibbonssoinne@aol.com}

\author[A. Stokes]{A. Stokes$^{2,3}$}

\address{$^{2}$Graduate School of Mathematical Sciences, The University of Tokyo, 3-8-1 Komaba Meguro-ku Tokyo 153--8914, Japan.}
\address{$^{3}$Faculty of Mathematics, Informatics and Mechanics, University of Warsaw, ul. Banacha 2, 02-097, Warsaw, Poland.}

\email{stokes@mimuw.edu.pl} 

\author[A. P. Veselov]{A. P. Veselov$^{4}$}

\address{$^{4}$Department of Mathematical Sciences, Loughborough University, Loughborough LE11 3TU, UK.}
\email{A.P.Veselov@lboro.ac.uk}


\maketitle

\begin{abstract}
	We study a delay-differential analogue of the first Painlev\'e equation obtained as a delay periodic reduction of Shabat's dressing chain. 
	We construct formal entire solutions to this equation and introduce a new family of polynomials (called Bernoulli-Catalan polynomials), which are defined by a nonlinear recurrence of Catalan type, and which share properties with Bernoulli and Euler polynomials.
We also discuss meromorphic solutions and describe the singularity structure of this delay Painlev\'e-I equation in terms of an affine Weyl group of type $A_1^{(1)}$.
As an application we demonstrate the link with the problem of calculation of the Masur-Veech volumes of the moduli spaces of meromorphic quadratic differentials by re-deriving some of the known formulas. 
\end{abstract}

\section{Introduction}
In 1992 Shabat \cite{Shabat} considered the following remarkable dynamical system, which he called dressing chain: 
\begin{equation}
f_{k+1}' + f_k' = f_{k+1}^2 - f_k^2 + \alpha_k, \qquad k \in \Z, \quad f_k = f_k(z), \quad \alpha_k \in \C
\end{equation}
and used it to construct new one-dimensional Schr\"odinger operators with explicit spectrum.
Soon after, Shabat and Veselov  \cite{VS} studied the periodic closures with $f_{k+N(z)}=f_k(z),$ revealing deep relations with algebraic geometry, spectral theory and Painlev\'e equations. 
In particular, for $N=3$ and $\alpha=\sum_{k=1}^N\alpha_k \neq 0$ we have a system which is equivalent to the Painlev\`e-IV equation, and for general $N$ we have higher analogues of Painlev\`e-IV and Painlev\`e-V equations \cite{VS,Adler}.
Shabat studied also the closure of the dressing chain with the self-similar condition
$f_{j+1}(z)=q f_j(qz)$
(see Section 3 in \cite{Shabat}).

Soon after that, Gibbons and Veselov \cite{GV} began to study the solutions of the delay-periodic closure of the dressing chain assuming that
\begin{equation}
f_{k+1} = T f_k, \qquad T \varphi(z) = \varphi(z+h),
\end{equation}
which leads to the delay-differential equation
\begin{equation} \label{delayP1}
(T+1) f' = (T-1) f^2 + \alpha.
\end{equation}
This equation is also a special case of a delay-differential equation identified by Grammaticos, Ramani and Moreira \cite{GRM93} via a kind of singularity confinement criterion as an analogue of the 
first Painlev\'e equation. 
Indeed, if we substitute
\begin{equation}
f(z)= \frac{1}{h} + h w(z), \qquad \alpha = \frac{h^3}{6}, 
\end{equation}
in equation \eqref{delayP1}, in the limit $h \to 0$ we have 
\begin{equation}
w_{z z z} = 12 w w_{z} + 1, 
\end{equation}
which after integrating once leads without loss of generality to the first Painlev\'e equation in the form
\begin{equation}
w_{z z  } = 6 w^2 + z.
\end{equation}
For this reason we will call equation \eqref{delayP1} the {\it delay Painlev\'e-I equation.}
The $\alpha=0$ case of this equation was shown by Berntson \cite{Ber18} to admit a multi-parameter family of elliptic function solutions, in parallel with autonomous limits of the Painlev\'e equations being solved by elliptic functions.

The equation \eqref{delayP1} can also be derived from the intermediate long wave (ILW) equation in the form
\begin{equation}
u_t = 2 u u_x - \frac{T+1}{T-1} u_{x x},
\end{equation}
where $Tu(t,x) = u(t,x+h)$ is the shift in $x$-direction by $h=iD$, $D$ being the depth of the water.
This important equation was written in this form by Chen and Lee (see formula (8) in \cite{Chen_Lee}), but in an equivalent form it appeared earlier in the paper \cite{Joseph} by Joseph. It interpolates between the famous Korteweg-de Vries (shallow water) and Benjamin-Ono (deep water) equations and has been extensively studied starting from \cite{SAK} (see the references in the recent survey \cite{KS21}).
Substituting in this equation
$$u(x,t) = -\frac{\alpha}{h} t + f(x -\frac{\alpha}{h} t^2)$$ 
and applying to both sides the operator $T-1,$ we see that $f$ must satisfy
equation \eqref{delayP1}.


The results of Gibbons and Kupershmidt \cite{GK} provide the following Lax pair for the equation \eqref{delayP1}:
\begin{subequations}
\begin{gather}
(\partial_x+T+u)\psi=-\frac{\alpha}{h}\lambda \psi, \\
\partial_\lambda \psi=((T+u)^2 - v)\psi, \quad v=\frac{T+1}{T-1}\partial_x u.
\end{gather}
\end{subequations}
In this paper we will study the solutions of the delay Painlev\'e-I  equation \eqref{delayP1}, continuing unfinished work \cite{GV}. Without loss of generality we will assume that the shift parameter $h=1.$

One of the main outcomes are new families of interesting polynomials, sharing some properties with the famous Bernoulli and Euler polynomials. In contrast to the classical case their generating function is not an elementary function but a (formal) solution of the delay Painlev\'e-I equation. A particularly interesting case, corresponding to special initial data, is given by the family of polynomials $Q_n(z)$, which we call {\it Bernoulli-Catalan polynomials}. They satisfy the following nonlinear recurrence relation of Catalan type:
\begin{equation}
Q_n(z)=\mathcal K \sum_{k=1}^{n-1}Q_k(z)Q_{n-k}(z),\quad n>1, \quad  Q_1(z):=z-\frac{1}{2},
\end{equation}
with
$$
\mathcal K= K(\tfrac{d}{dz}), \quad K(t)=\frac{1}{t}\tanh\frac{t}{2}=\sum_{n=1}^\infty\frac{2^{2n}-1}{(2n)!}B_{2n}t^{2n-2},
$$
where $B_{2n}$ are the classical Bernoulli numbers. 

We demonstrate their importance by linking them with the calculation of the Masur-Veech volumes, using the recent results of Yang, Zagier and Zhang \cite{YZZ}.
The Masur-Veech volume 
$\operatorname{Vol} \mathcal Q_{g,n}$ 
measures the moduli space $\mathcal Q_{g,n}$ of algebraic curves of genus $g$ with $n$ marked points, supplied with a meromorphic quadratic differential having simple poles at these points. 
Yang, Zagier and Zhang \cite{YZZ} derived the equations for a corresponding generating function for the family of $\operatorname{Vol} \mathcal Q_{g,n}$ as $g$ and $n$ vary, one of which is closely related to our equation \eqref{delayP1} (see Section \ref{sec:masurveech}). This gives an opportunity to apply our results to the problem of the explicit calculation of these volumes, which has been the subject of active research in recent years (see e.g. \cite{ADGZZ} and references therein).

We start the paper with a conservation form of the delay Painlev\'e-I equation in Section \ref{sec:conservationform}, then in Section \ref{sec:existence} we prove the existence of a large class of formal entire solutions via a recursive procedure. 

In Sections \ref{sec:polynomials} and \ref{sec:alls_izero} we study the polynomials which appear in the case of zero initial seed. In the particular special case with zero parameters we have the new family of Bernoulli-Catalan polynomials. 

In Section \ref{sec:meromorphicsols} we discuss the meromorphic solutions and the singularity structure of the delay Painlev\'e-I equation, turned out to be related with an action of the affine Weyl group of type $A_1^{(1)}$.

Finally in Section \ref{sec:masurveech} we establish the link with the calculation of the Masur-Veech volumes, providing, in particular, one more derivation of the Kontsevich formula for the Masur-Veech volume of $\mathcal Q_{0,n.}$ 
We conclude with the discussion of some open questions and the relation with the discrete Painlev\'e-I equation.

\section{The conservation form of the equation} \label{sec:conservationform}

Assume from now on that the shift parameter $h=1$, so  $T \varphi(z) = \varphi(z+1)$. 
We start with the following simple observation.

\begin{proposition}
The delay Painlev\'e-I equation \eqref{delayP1} can be rewritten in the conservation form 
\begin{equation}
\label{conserved1}
\frac{d}{dz}S=0,
\end{equation}
where
\begin{equation}
\label{conserved2}
S:=(T+1)f(z)-\int_{z}^{z+1} [f^2(t)+\alpha(t- \tfrac{1}{2})]dt.
\end{equation}
In other words, $S$ is a conserved quantity of the delay Painlev\'e-I  equation.
\end{proposition}

Indeed,
$$
\frac{d}{dz}S(z)= (T+1)\frac{d}{dz}f(z)-(T-1) f^2 (z)-\alpha=0
$$
if and only if $f$ satisfies \eqref{delayP1}. The shift by $1/2$ in the integrand is chosen here for convenience.

We will use this form of the equation and the conserved quantity $S=S(\alpha)$ in Section \ref{sec:polynomials}.

\section{Formal entire solutions} \label{sec:existence}

Consider first the case of equation \eqref{delayP1} when parameter $\alpha = 0$:
\begin{equation} \label{alpha0}
(T+1)f' = (T-1) f^2, \qquad T \varphi(z) = \varphi(z+1).
\end{equation}
Let $V_0$ be the vector space of entire 2-periodic functions:
\begin{equation}
\left. (T^2 - 1) \right|_{V_0} = 0.
\end{equation}
We have the spectral decomposition for $T$
\begin{equation}
V_0 = V_0^{+} \oplus V_0^{-}, \qquad
V_0^{\pm} = \left\{ \varphi \in V_0 ~:~ T \varphi = \pm \varphi \right\}.
\end{equation}
Then any function $f \in V_0^{-}$ satisfies \eqref{alpha0}. The question of whether there are other entire solutions is interesting, but we will not discuss this here. 

Consider now the following solutions of the delay Painlev\'e-I  equation
\begin{equation} \label{1}
(T+1) f' = (T-1) f^2 +  \alpha, \qquad T \varphi(z) = \varphi(z+1),
\end{equation}
with $\alpha\neq 0$, as a formal series in $\alpha$:
\begin{equation} \label{2}
f(z) = \phi_0 (z) + \sum_{k = 1}^{\infty} \phi_k(z)  \alpha^k, \qquad \phi_0 \in V_0^{-}.
\end{equation}
Our first main result is that such formal solutions do exist, with each $\phi_k$ being an entire function of $z$. 
Substitution of \eqref{2} into \eqref{1} gives $(T+1)\phi_0' = (T-1) \phi_0^2$, which is satisfied automatically, then we have recursion relations for the functions
\begin{equation}
\psi_n (z) \defeq \phi_{n}(z+1) + \phi_n (z) = (T+1) \phi_n (z), 
\end{equation}
which take the form 
\begin{subequations}
\begin{align}
&\psi_1 '  + 2 \phi_0 (z) \psi_1 - 1 = 0, \label{psi1} \\
&\psi_n' + 2 \phi_0 (z) \psi_n + (1-T)\sum_{i=1}^{n-1} \phi_i(z) \phi_{n-i}(z) = 0, \quad n \geq 2. \label{psin} 
\end{align}
\end{subequations}
The general solution of \eqref{psi1} is
\begin{equation}
\psi_1(z) = C(z) e^{\rho(z)}, \qquad \rho(z) \defeq  - 2 \int^z \phi_0(\xi) d \xi, \qquad C(z) \defeq \int^z e^{ - \rho(\xi) d \xi}.
\end{equation}
Note that $\rho(z)$ is 2-periodic since $\int_0^2 \phi_0(\xi) d \xi = 0$. 
Since $C(z)$ is a sum of linear and 2-periodic functions and $e^{\rho(z)}$ is purely 2-periodic, we introduce the vector space of entire functions
\begin{equation}
V_1 = \left\{ z p_1 (z) +p_0 (z) ~:~ p_{0,1} \in V_0 \right\},
\end{equation}
and see that $\psi_1 \in V_1$.
To find $\phi_1$ we need to solve the equation
\begin{equation} \label{psiphi1}
(T+1) \phi_1 = \psi_1.
\end{equation}
The solution may not exist in $V_1$, but we claim that it always exists in 
\begin{equation}
V_2 = \left\{ z^2 p_2(z) + z p_1 (z) + p_0(z) ~:~ p_i \in V_0 \right\}.
\end{equation}
We will prove a more general fact. Let 
\begin{equation}
V_n = \left\{ z^n p_i(z) +\dots + p_0(z) ~:~ p_i \in V_0 \right\},
\end{equation}
be the vector space of `quasi-polynomials' of degree less than or equal to $n$. 
\begin{lemma} \label{lemmaquasipolynomial}
For every $\psi \in V_n$, there exists $\phi \in V_{n+1}$ such that $(T+1)\phi = \psi$.
\end{lemma}
\begin{proof}
We need to show that any $\psi = p_n z^n +\dots+ z p_1 + p_0$ with $p_k \in V_0$ can be realised as $(T+1)\phi$.
Let $\phi = z^{n+1} q_{n+1} + \dots + z q_1 + q_0 \in V_{n+1}$, with $q_{n+1} \in V_0^{-}$.
Direct computation shows that
 \begin{equation}
(T+1) \phi = (T+1) \sum_{i=0}^{n+1}  z^i q_{i} = \sum_{k=0}^n z^k \left\{ -\binom{n+1}{k} q_{n+1} + \sum_{i=k}^{n} \binom{i}{k} Tq_i + q_k \right\}.
\end{equation}
Split the remaining 2-periodic coefficients in $\phi$ according to $q_k = q_k^{+} + q_k^{-}$, where $q_k^{\pm} \in V_0^{\pm}$. Explicitly, $q_k^{+} = \frac{1}{2} (q_k + T q_k)$,  $q_k^{-} = \frac{1}{2} (q_k - T q_k)$. Then 
 \begin{equation*}
(T+1) \phi =z^n \left( -(n+1) q_{n+1} + 2 q_n^+ \right) + \sum_{k=0}^{n-1} z^k \left\{ -\binom{n+1}{k} q_{n+1} + \sum_{i=k+1}^{n} 2 \binom{i}{k} q_i^+ + 3q_k^+ + q_k^- \right\}.
\end{equation*}
To obtain $\psi = \sum_{k=0}^{n} z^k p_k$ as $(T+1) \phi$, we again decompose $p_k = p_k^+ + p_k^-$, $p_k^{\pm} \in V_0^{\pm}$, so we require
\begin{subequations}
\begin{align}
p_n &= p_n^{+} + p_n^- =  -(n+1) q_{n+1} + 2 q_n^+, \\
p_k &= p_k^{+} + p_k^- =  \binom{n+1}{k} q_{n+1} + \sum_{i=k-1}^n  2\binom{i}{k} q_i^+ + 3q_k^+ + q_k^-,
\end{align}
\end{subequations}
for $k=0,\dots, n-1$. Noting that $q_{n+1} \in V_0^{-}$, we can choose
\begin{equation}
q_{n+1} = - \frac{p_n^-}{n+1}, \qquad q_n^+ = \frac{p_n^+}{2},
\end{equation}
then the remaining $n$ equations are of the form $p_k = F( q_k^-, q_k^+, q_{k+1}^-, \dots , q_n^+, q_{n+1})$, so the system is triangular and may be solved for $q_k^+, q_k^-$. Note that for $\psi$ such that $p_n \in V_0^+$, $q_{n+1}$ can be chosen to be zero and therefore in such case $\phi \in V_n$.
\end{proof}

Summarising all this, we have the following main result of this section.

\begin{theorem}
For every function $\phi_0(z) \in V_0^{-}$ there exists a formal series solution of the delay Painlev\'e-I  equation \eqref{1}
of the form
\begin{equation}
f(z) = \phi_0(z) + \sum_{k=1}^{\infty} \phi_k(z)  \alpha^k,
\end{equation}
with quasi-polynomial entire functions $\phi_k(z) \in V_{3k-1}$.
\end{theorem}

\begin{proof}
By induction. If $\phi_k \in V_{3k-1}$ for $k\leq n-1$, then $\sum_{i=1}^{n-1} \phi_i \phi_{n-i} \in V_{3n-2}$, so $(T-1)\sum_{i=1}^{n-1} \phi_i \phi_{n-i} \in V_{3n-3}$. Then solving the differential equation \eqref{psin} gives $\psi_n \in V_{3n-2}$, after which Lemma \ref{lemmaquasipolynomial} ensures that $\phi_n \in V_{3n-1}$ and the proof is complete.
\end{proof}

\section{Polynomials from the trivial seed}
\label{sec:polynomials}

The simplest choice of entire antiperiodic `seed' function $\phi_0(z)$ to determine the coefficients $\phi_n(z)$ by the recursive procedure of the previous Section is $\phi_0(z) = 0$. Moreover, this will lead to polynomial $\phi_n=\Phi_n$. 
In order to simplify some expressions later we make a slight scaling of the perturbation of the $\alpha=0$ case from Section \ref{sec:existence} and instead let 
$\alpha=2\beta$
 and consider
\begin{equation} \label{scaledperturbation}
(T+1) f' = (T-1) f^2 + 2 \beta,
\end{equation}
\begin{equation} \label{22}
f(z) = \sum_{n = 0}^{\infty} \Phi_n(z)  \beta^n, \quad \Phi_0(z) \equiv 0.
\end{equation}
In this case the equations corresponding to \eqref{psi1}, \eqref{psin} for the recursive procedure are
\begin{equation} \label{recurrencephi}
\Psi_1' = 2, \qquad 
\Psi_n' = (T-1) \sum_{i=1}^{n-1} \Phi_i \Phi_{n-i}, 
\qquad \Psi_n = (T+1)\Phi_n.
\end{equation}
Note that because $\Phi_0(z) \equiv 0$ one could include in the summation the zero term with $i=n$ to make the recurrence look more natural, but we prefer not to do this.

In decomposing $\Psi_n= (T+1)\Phi_n$, we note that $(T+1)$ is invertible on polynomials, and that the degree of $\Phi_n$ is $n$. At each step we perform an integration to compute $\Psi_n$, leaving the constant of integration 
$$\tau_n=\Psi_n(0)=\Phi_n(1)+\Phi_n(0)$$
arbitrary at this stage. 
The first few of the polynomials $\Phi_n=\Phi_n(z;\tau)$ are as follows:

\begin{align*}
\Phi_0(z;\tau)&=0,\\
\Phi_1(z;\tau)&= z + \tfrac{1}{2}(\tau_1 - 1), \\
\Phi_2(z;\tau) &= \tfrac{1}{2} z^2 + \tfrac{1}{2} ( \tau_1 - 1)z + \tfrac{1}{4}(-\tau_1 + 2 \tau_2), \\
\Phi_3(z;\tau) &= \tfrac{1}{2} z^3 + \tfrac{3}{4} ( \tau_1 - 1) z^2 + \tfrac{1}{4}(-3\tau_1+  \tau_1^2 +2 \tau_2) z + \tfrac{1}{8}(1- \tau_1^2 - 2\tau_2 + 4 \tau_3), \\
\Phi_4(z;\tau) &= \tfrac{5}{8} z^4 + \tfrac{5}{4} ( \tau_1 - 1) z^3 + \tfrac{1}{8}(-1 - 15\tau_1 +6 \tau_1^2 + 6 \tau_2) z^2 \\
&\quad + \tfrac{1}{8} ( \tau_1^3 - 6 \tau_1^2  - \tau_1 + 4 \tau_1 \tau_2 - 6 \tau_2 + 4\tau_3 + 6) z + \tfrac{1}{16}( - \tau_1^3 + 6 \tau_1 -4 \tau_1 \tau_2 - 4 \tau_3 + 8 \tau_4).
\end{align*}

Let the coefficients of the polynomials $\Phi_n(z;\tau)$ written in the standard basis of monomials in $z$ be $a_{0}^{(n)}, \dots, a_{n}^{(n)}$, so 
\begin{equation} \label{phicoeffs}
\Phi_n(z ; \tau) = \sum_{k=0}^n a^{(n)}_k z^k.
\end{equation}
It can be shown directly from the defining recurrence \eqref{recurrencephi} that the coefficients satisfy the following relations for $m=0,\dots, n-1$:
\begin{equation} \label{coeffrelations}
(m+1) \left[ 2 a_{m+1}^{(n)} + \sum_{i=m+2}^{n} \binom{i}{m+1} a_i^{(n)} \right] = \sum_{i=1}^{n-1} \sum_{l=m+1}^{n} 
\sum_{\underset{0\leq k \leq n-i}{\underset{0\leq j \leq i}{j+k=l}}}
 \binom{l}{m} a_j^{(i)} a_k^{(n-i)}. 
\end{equation}
The relations \eqref{coeffrelations} can be used to determine the coefficients of $\Phi_n$ in terms of those of $\Phi_1, \dots, \Phi_{n-1}$ as well as 
$\tau_n$.
The first property of these polynomials we will establish is the following.
\begin{proposition} \label{prop:generalsymmetry}
The polynomials $\Phi_n(z;\tau)$ satisfy the symmetry relation
 \begin{equation}
\Phi_n(1-z; \tau) = (-1)^n \Phi_n(z;\tilde \tau),
\end{equation}
where $\tilde{\tau}=(\tilde{\tau}_1,\dots, \tilde{\tau}_n)$ with $\tilde{\tau}_i = (-1)^{i}\tau_i.$
In particular, when all odd $\tau_{2l-1}=0$ the corresponding polynomials satisfy the symmetry property of the Bernoulli polynomials:
 \begin{equation}
 \label{sym}
\Phi_n(1-z) = (-1)^n \Phi_n(z).
\end{equation} 
\end{proposition}
\begin{proof}
By induction. Assume that the symmetry holds for all $n < N$ and consider $\tilde{\Psi}_N(z) = (T+1)\Phi_N(z; \tilde{\tau})$. Under this inductive hypothesis it can be shown directly from the recurrence \eqref{recurrencephi} that 
\begin{equation}
\tilde{\Psi}_N ' = (-1)^{N+1} \Psi_N',
\end{equation}
so after integrating once and noting that $\tau_N = \Psi_N(0)$, we have the result.
\end{proof}

The leading coefficients of the polynomials $\Phi_n(z;\tau)$ can be expressed in terms of the 
famous {\it Catalan numbers}
\begin{equation}\label{Catalan}
C_n:=\frac{1}{n+1}\binom{2n}{n}, \quad  \, n=0,1,2,\dots,
\end{equation}
which play a very prominent role in combinatorics \cite{Stanley}:
$$
C_n=1,1,2,5,14,42,132,429,1430,4862, 16796, 58786, \dots
$$
They can be defined also by the recurrence relation
\begin{equation}\label{recCat}
C_n=\sum_{i=1}^n C_{i-1}C_{n-i},\quad C_0=1,
\end{equation}
or by the generating function
\begin{equation}\label{genCat}
\sum_{n=0}^\infty C_nx^n=\frac{1-\sqrt{1-4x}}{2x}.
\end{equation}

\begin{proposition} \label{prop:leadingcoeffs}
\begin{enumerate}[(a)] 
\item The highest coefficient $A_n:= a_n^{(n)}$ of the polynomial $\Phi_n(z;\tau), \, n\geq 1$ is 
\begin{equation}\label{ann}
A_n= \frac{(2n-3)!!}{n!}=\frac{(2n-2)!}{2^{n-1} n! (n-1)!}=2^{1-n}C_{n-1},
\end{equation}
where $C_n$ is the corresponding Catalan number.
\item For $n\geq1$, the coefficient of $z^{n-1}$ in the polynomial $\Phi_n(z;\tau)$ is given by 
\begin{equation}
a_{n-1}^{(n)} =  \frac{n}{2}A_n (\tau_1 - 1) =  n2^{-n}C_{n-1} (\tau_1 - 1).
\end{equation}
\item When $\tau_1=0$ the coefficient of $z^{n-2}$ in the polynomial $\Phi_n(z;\tau)$ equals 
\begin{equation}
a_{n-2}^{(n)} =  (n-1)A_{n-1}\frac{\tau_2}{2}+\frac{n(n-1)}{8}A_n-\frac{(n-1)}{24}A_{n-1}-\frac{2^{n-2}(n-1)}{12}.
\end{equation}
For general $\tau_1$ we have for $n\geq2$ the formula
\begin{equation}
a_{n-2}^{(n)} =  (n-1)A_{n-1}\frac{\tau_2}{2}+\frac{n(n-1)}{8}A_n-\frac{(n-1)}{24}A_{n-1}-\frac{2^{n-2}(n-1)}{12} + L_n \tau_1^2 - L_{n+1} \tau_1,
\end{equation}
where $L_n$ are rational numbers defined by the recurrence
\begin{equation}
L_n = 2^{n-5} + \frac{1}{2}\sum_{i=1}^{n-1} \left( A_{n-i} L_i + A_i L_{n-i}\right), \qquad L_1=L_2=0,
\end{equation} 
so in particular $L_3=\frac{1}{4}$, $L_4=\frac{3}{4}$, $L_5=\frac{15}{8}$, $L_6=\frac{35}{8}$, $L_7=\frac{315}{32}$, $\dots$ 
\end{enumerate}
\end{proposition}

The above can be proved inductively using the recurrence relations \eqref{coeffrelations}, implying, in particular, that
$$
A_n = \frac{1}{2} \sum_{i=1}^{n-1} A_i A_{n-i},  \quad A_{n-1} = \frac{1}{n-1}\sum_{i=1}^{n-1} (n-1-i) A_i A_{n-1-i}.
$$


\subsection{Conserved quantity and parameters $\tau$}

From Section \ref{sec:conservationform}, the equation \eqref{scaledperturbation} has the conserved quantity \eqref{conserved2}:
\begin{equation}
\label{cons3}
S(z,\beta):=(T+1)f(z)-\int_z^{z+1} [f^2(t)+2\beta(t-\tfrac{1}{2})]dt,
\end{equation}
which depends only on $\beta$ but not on $z$:
\begin{equation}
\label{sk}
S(z,\beta)= S(\beta)=\sum_{k=1}^\infty s_k \beta^k.
\end{equation}
In particular, setting $z=0$ we have the relation
$$
S(\beta)=f(1)+f(0)-\int_0^{1} [f^2(t)+2\beta(t-\tfrac{1}{2})]dt =\sum_{k=1}^\infty s_k \beta^k.
$$
Substituting here the expansion (\ref{22}) and using the definition $\tau_k=\Psi_k(0)=\Phi_k(1)+\Phi_k(0)$ we have
\begin{equation}
\label{stau}
\sum_{k=1}^\infty \tau_k \beta^k-\int_0^1\left(\sum_{n = 1}^{\infty} \Phi_n(z; \tau)  \beta^n\right)^2dz=\sum_{k=1}^\infty s_k \beta^k.
\end{equation}
This allows us to express the integration constants $s_k$ via initial value parameters $\tau_j$.

\begin{proposition}
The integration constants $s_k$ are related to the parameters $\tau_j$ by a triangular polynomial transformation
\begin{equation} \label{change2}
s_n=\tau_n+ S_n(\tau_1,\dots,\tau_{n-1}), \quad S_n \in \mathbb Q[\tau_1,\dots,\tau_{n-1}].
\end{equation}
Polynomials $S_n(\tau)$ have the symmetry property
 \begin{equation}
 S_n(\tilde\tau) = (-1)^n S_n(\tau),
\end{equation}
where $\tilde{\tau}_i = (-1)^{i}\tau_i.$ In particular, $S_{2k-1}(0)=0$ for all $k\in \mathbb N.$
\end{proposition}
Here is the explicit form of $s_k$ in terms of parameters $\tau$ with $k\leq 7$:
\begin{equation} \label{change3}
\begin{aligned}
s_1&=\tau_1, \\
s_2&=\tau_2-\tfrac{\tau _1^2}{4}-\tfrac{1}{12}, \\
s_3&=\tau _3 - \tfrac{\tau _1 \tau _2}{2}, \\
s_4&=\tau_4+\tfrac{\tau _1^2}{16}-\tfrac{\tau _2^2}{4}-\tfrac{\tau _1 \tau _3}{2}+\tfrac{1}{24}, \\
s_5&=\tau _5+\tfrac{\tau _1^3}{12}+\tfrac{\tau _1 \tau _2}{8}-\tfrac{\tau _2 \tau _3}{2}-\tfrac{\tau _1 \tau _4}{2},\\
s_6&=\tau _6+\tfrac{5 \tau _1^4}{64}-\tfrac{\tau _3^2}{4}+\tfrac{\tau _2^2}{16}+\tfrac{ \tau _1^2\tau _2}{4} -\tfrac{7 \tau _1^2}{24}-\tfrac{\tau _2 \tau _4}{2}+\tfrac{\tau _1 \tau _3}{8}-\tfrac{\tau _1 \tau _5}{2}-\tfrac{41}{192},\\
s_7&= \tau_7 + \tfrac{\tau _1^5}{16}+\tfrac{5\tau _1^3\tau _2 }{16} -\tfrac{7 \tau _1^3}{8}+\tfrac{\tau _1^2\tau _3 }{4} +\tfrac{\tau _1\tau _2^2 }{4} +\tfrac{ \tau _1 \tau _4}{8}-\tfrac{7  \tau _1\tau _2}{12}+\tfrac{\tau _2 \tau _3}{8}-\tfrac{\tau _3 \tau _4}{2}-\tfrac{\tau _2 \tau _5}{2}-\tfrac{\tau _1\tau _6 }{2}.
\end{aligned}
\end{equation}
We will present the formulas expressing $\tau_i$ in terms of $s_k$ at the end of Section 4.

\subsection{Polynomials $\Phi_n(z)$ with all $\tau_i=0$}


There are two natural special cases of polynomials $\Phi_n$, corresponding to zero values of parameters $\tau_i$ and integration constants $s_k$ respectively.
We start first with the case
\begin{equation}
\Phi_n(z) \defeq \Phi_n(z; 0),
\end{equation}
where all $\tau_i$ are set to zero, leaving the case with $s_k=0$ for Section \ref{sec:alls_izero}.

The first few of the $\Phi_n(z)$ have the form
\begin{align*}
\Phi_0(z)&=0,\\
\Phi_1(z) &= z - \tfrac{1}{2}, \\
\Phi_2(z) &= \tfrac{1}{2} z^2 - \tfrac{1}{2}z, \\
\Phi_3(z) &= \tfrac{1}{2} z^3 - \tfrac{3}{4} z^2 + \tfrac{1}{8},\\
\Phi_4(z) &= \tfrac{5}{8} z^4 - \tfrac{5}{4} z^3 - \tfrac{1}{8}z^2+\tfrac{3}{4}z,\\
\Phi_5(z) &= \tfrac{7}{8}z^5 - \tfrac{35}{16} z^4 -\tfrac{7}{12} z^3 + \tfrac{49}{16} z^2 - \tfrac{7}{12},\\
\Phi_6(z) &= \tfrac{21}{16} z^6 - \tfrac{63}{16} z^5 - \tfrac{185}{96}z^4 + \tfrac{125}{12} z^3 + \tfrac{7}{12}z^2 - \tfrac{619}{96} z,\\
\Phi_7(z) &= \tfrac{33}{16} z^7 - \tfrac{231}{32} z^6 - \tfrac{11}{2} z^5 + \tfrac{2035}{64} z^4 + \tfrac{451}{96}z^3 - \tfrac{2717}{64} z^2 + \tfrac{1595}{192},\\
\Phi_8(z) &= \tfrac{429}{128} z^8 - \tfrac{429}{32} z^7 - \tfrac{2779}{192} z^6 + \tfrac{2891}{32} z^5 + \tfrac{9365}{384} z^4 - \tfrac{20639}{96} z^3 - \tfrac{1595}{192} z^2 + \tfrac{4259}{32} z.
\end{align*}

We present some plots of the polynomials $\Phi_n(z)$ for real $z$ in Figure \ref{fig:phiplots2}.

\begin{figure}[htb]
\begin{subfigure}{0.43\textwidth}
\includegraphics[width=\textwidth]{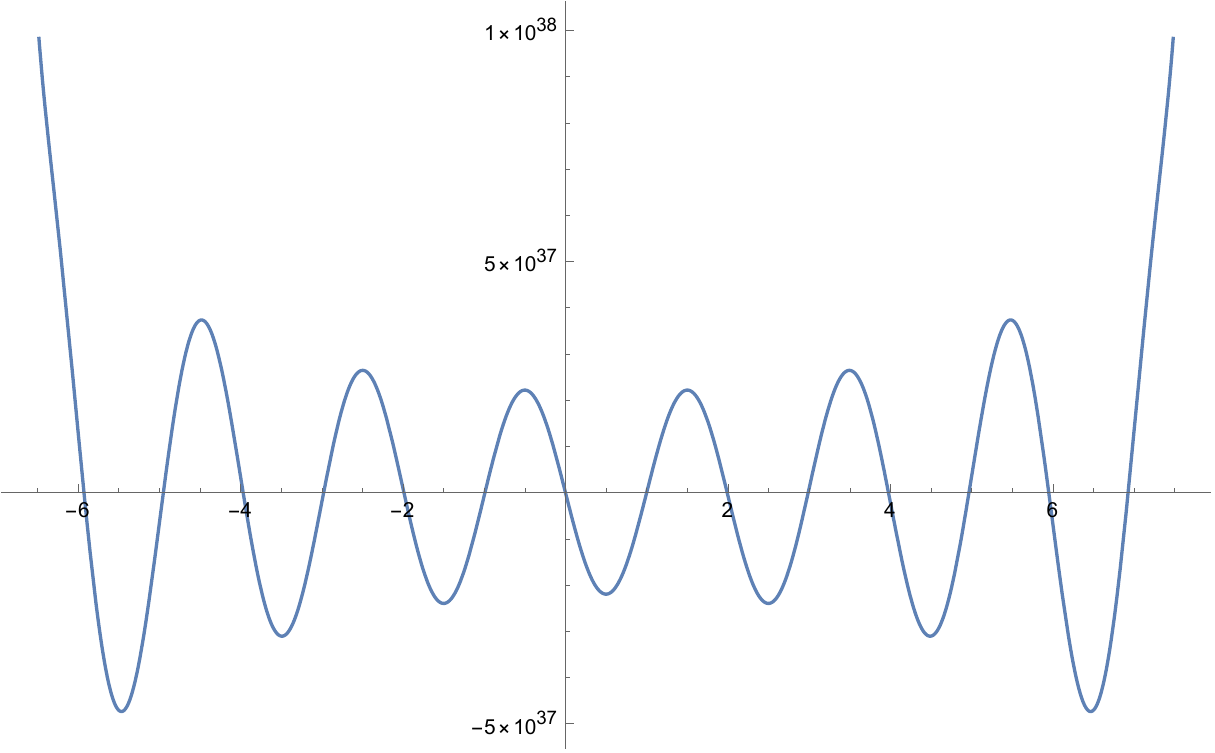}
\caption{$\Phi_{38}(z)$}
\label{fig:subim11}
\end{subfigure}
\quad
\begin{subfigure}{0.43\textwidth}
\includegraphics[width=\textwidth]{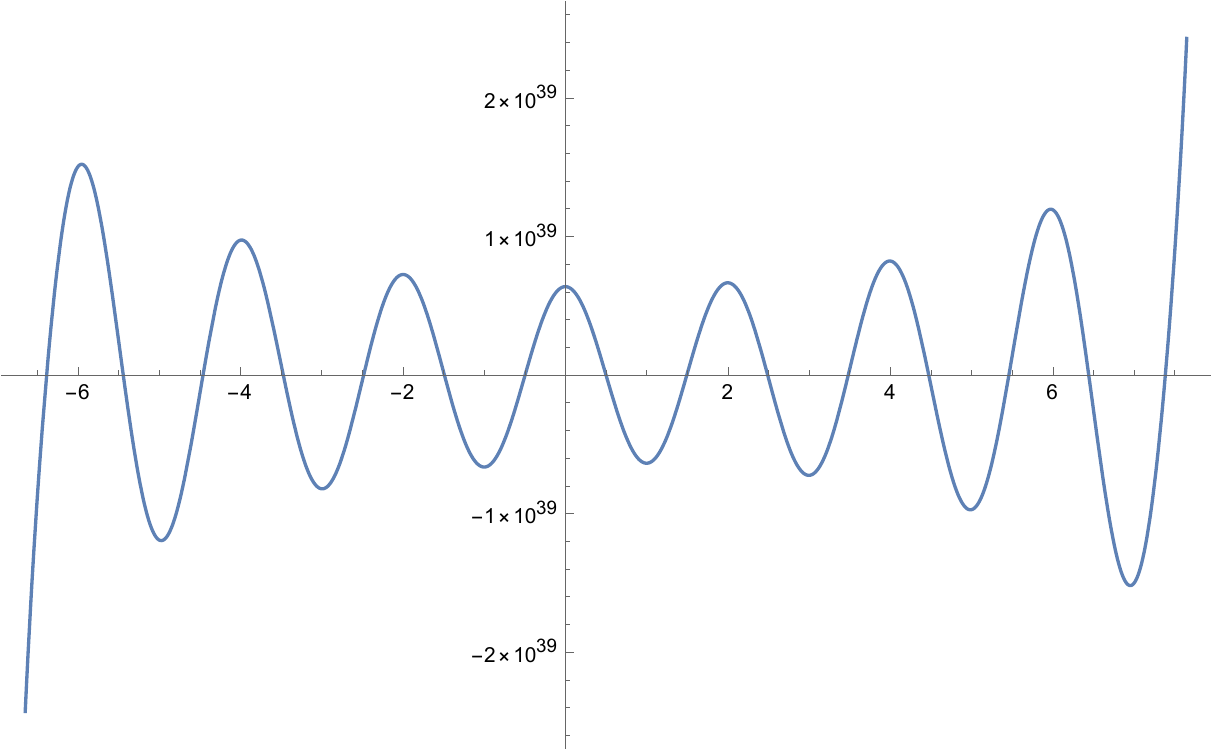}
\caption{$\Phi_{39}(z)$}
\label{fig:subim12}
\end{subfigure}
\\
\begin{subfigure}{0.43\textwidth}
\includegraphics[width=\textwidth]{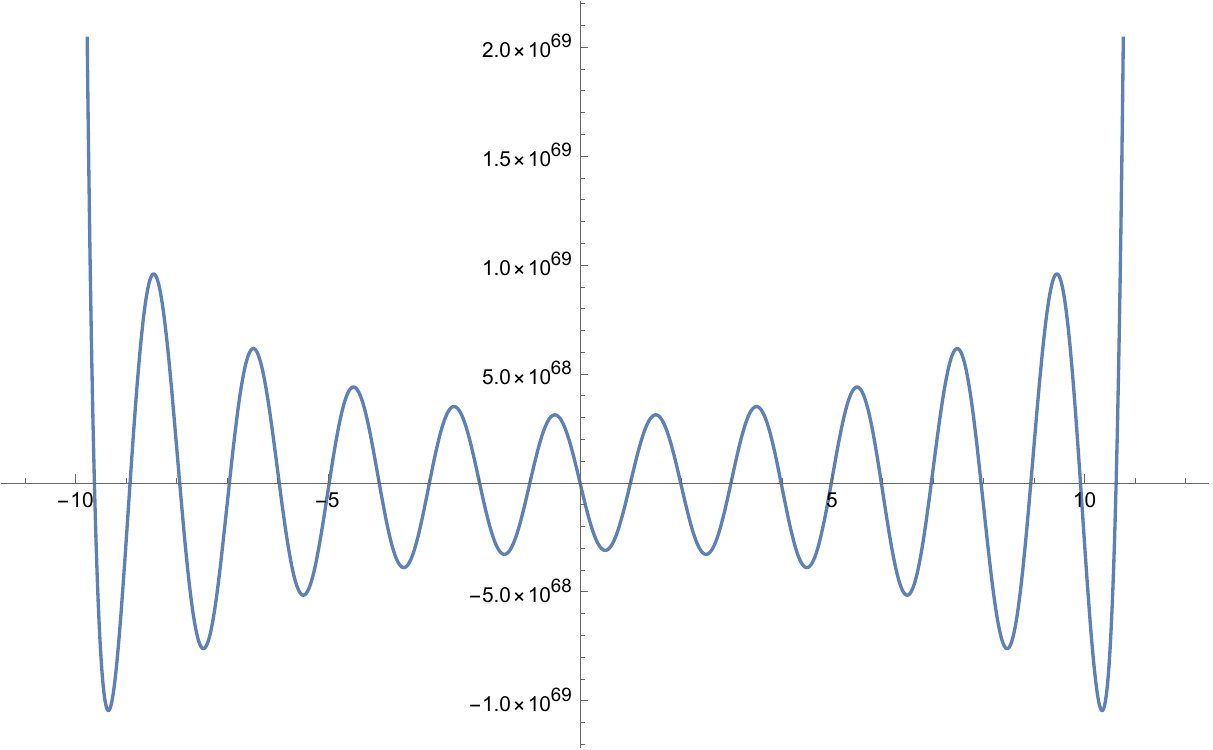}
\caption{$\Phi_{58}(z)$}
\label{fig:subim13}
\end{subfigure}
\quad
\begin{subfigure}{0.43\textwidth}
\includegraphics[width=\textwidth]{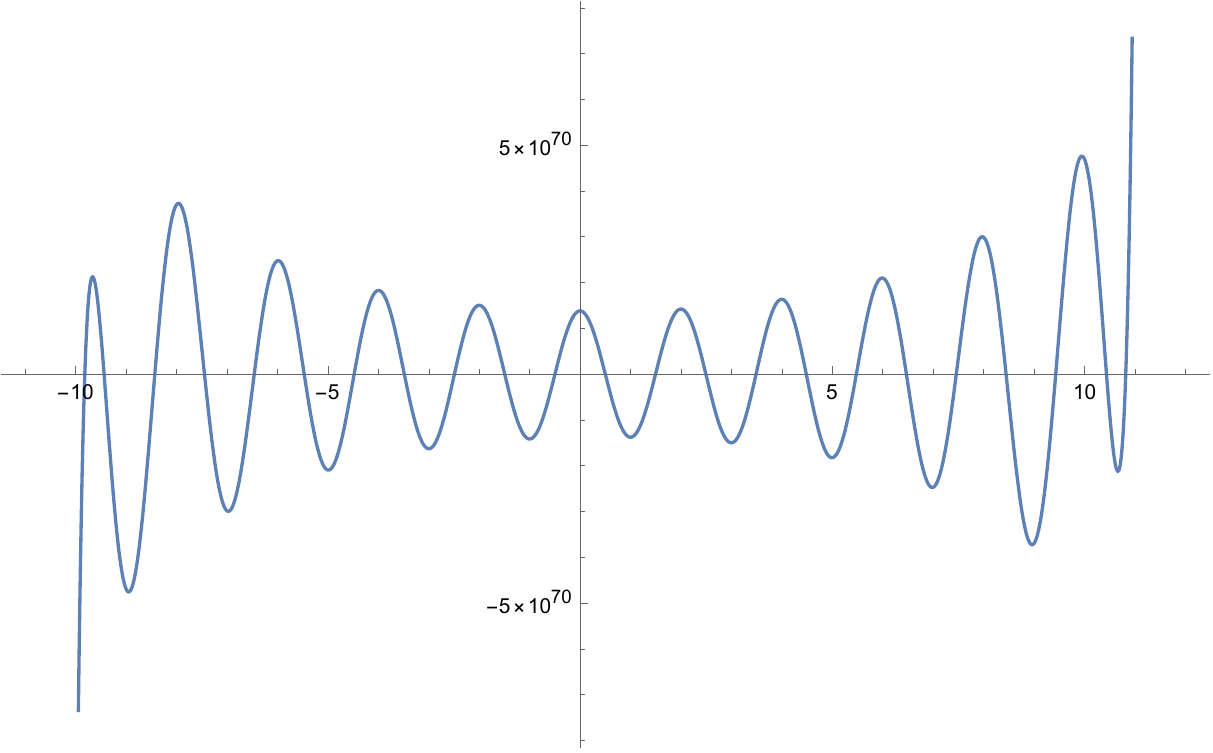}
\caption{$\Phi_{59}(z)$}
\label{fig:subim14}
\end{subfigure}
\caption{Plots of $\Phi_n(z)$ for real $z$}
\label{fig:phiplots2}
\end{figure}

\bigskip
The following proposition provides some information about the coefficients of these polynomials, the proof of which easily follows from the definition of $\tau_k=\Phi_{k}(1) + \Phi_{k}(0)$ and the symmetry property (\ref{sym}).

\begin{proposition}
For $k\geq1$, we have 
\begin{enumerate}[(a)]
\item  $\Phi_{2k}(0) = a_0^{(2k)}= \Phi_{2k}(1) = \sum_{i=0}^{2k} a_i^{(2k)}= 0$,
\item  $\Phi_{2k+1}(\tfrac{1}{2}) = \sum_{i=0}^{2k+1} 2^{-i} a_i^{(2k)} = 0$.
\item $\Phi_{2k+1}'(0)=a_1^{(2k+1)} = 0$.
\end{enumerate}
\end{proposition}

We believe  that the signs of $\Phi_{2k+1}(0)$ are alternating, but we cannot yet prove this (see, however, Proposition \ref{prop:alternatingsigns} below).



We would like to note that, despite the highly nonlinear nature of the recurrence \eqref{recurrencephi} defining $\Phi_n(z; \tau)$, in addition to the symmetry (\ref{sym}) there are more intriguing similarities between the properties of $\Phi_n(z;\tau)$ with $\tau_{2k+1}=0$ and of the Bernoulli polynomials
\begin{equation}
B_n(x) \defeq \frac{\partial}{T-1} x^n= \sum_{k=0}^{n}{n \choose k} B_k x^{n-k}, \quad n\geq 0:
\end{equation}
$$
B_0=1, \,\,\,\, B_1=x-\frac{1}{2},  \,\,\,\, B_2=x^2-x+\frac{1}{6},  \,\,\,\, B_3=x^3-\frac{3}{2}x^2+\frac{1}{2}x,  \,\,\,\, B_4=x^4-2x^3+x^2-\frac{1}{30}, 
$$
$$
B_5=x^5-\frac{5}{2}x^4+\frac{5}{3}x^3-\frac{1}{6}x, \,\,\,\, B_6=x^6-3x^5+\frac{5}{2}x^4-\frac{1}{2}x^2+\frac{1}{42}, \dots
$$
and the closely related Euler polynomials
\begin{equation}
E_n(x) \defeq \frac{2}{T + 1} x^n=\frac{2}{n+1}[B_{n+1}(x)-2^{n+1}B_{n+1}(\tfrac{x}{2})]:
\end{equation}
$$
E_0=1, \,\,\,\, E_1=x-\frac{1}{2},\,\,\,\, E_2=x^2-x, \,\,\,\, E_3=x^3-\frac{3}{2}x^2+\frac{1}{4}, \,\,\,\, E_4=x^4-2x^3+x, $$
$$
E_5=x^5-\frac{5}{2}x^4+\frac{5}{2}x^2-\frac{1}{2}, \,\,\,\, E_6=x^6-3x^5+5x^3-3x, \dots.
$$


For example, the polynomials $\Phi_n(z)$ seem to approximate trigonometric functions on some interval centred around $z=1/2$, but with varying amplitude (see Figure \ref{fig:phiplots2}) as opposed to the Bernoulli and Euler polynomials \cite{dilcher1, VW2} (see Figures \ref{fig:bernoulliplots} and \ref{fig:eulerplots}).
\begin{figure}[htb]
\centering
\begin{subfigure}{0.45\textwidth}
\includegraphics[width=\textwidth]{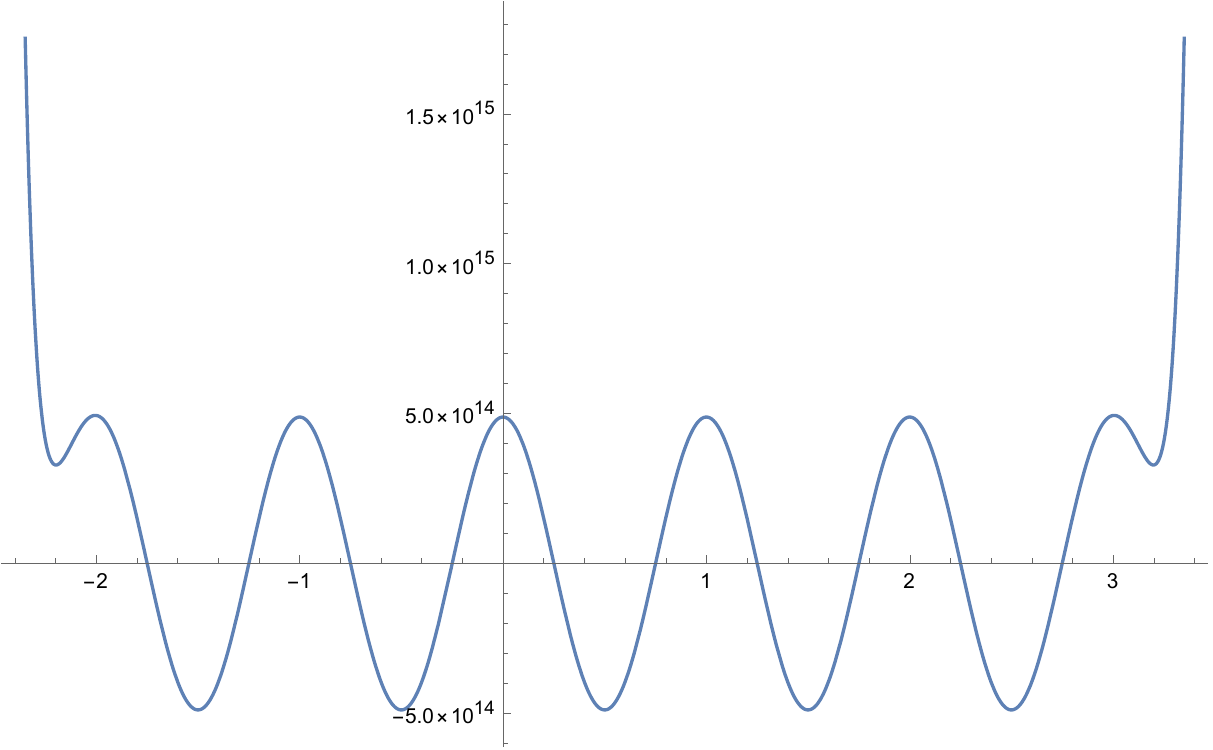}
\caption{$B_{38}(x)$}
\end{subfigure}
\quad
\begin{subfigure}{0.45\textwidth}
\includegraphics[width=\textwidth]{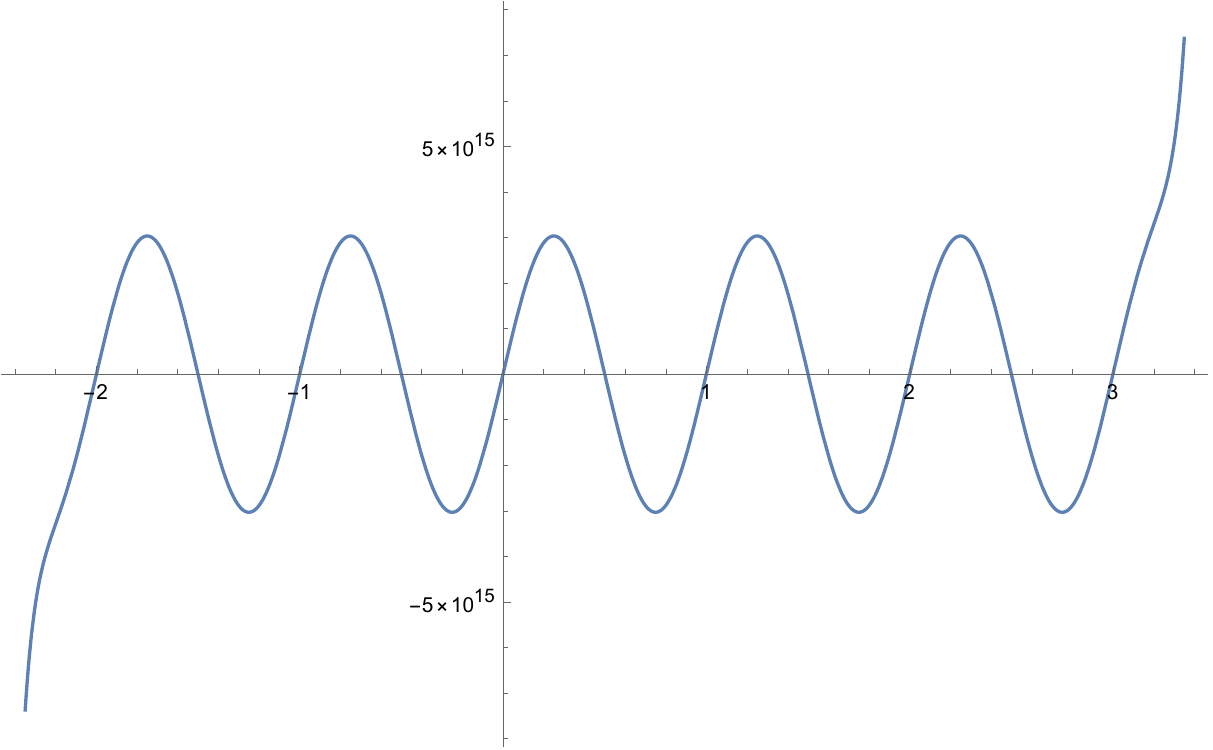}
\caption{$B_{39}(x)$}
\end{subfigure}
\caption{Plots of Bernoulli polynomials $B_n(x)$ for real $x$}
\label{fig:bernoulliplots}
\end{figure}

\begin{figure}[htb]
\centering
\begin{subfigure}{0.45\textwidth}
\includegraphics[width=\textwidth]{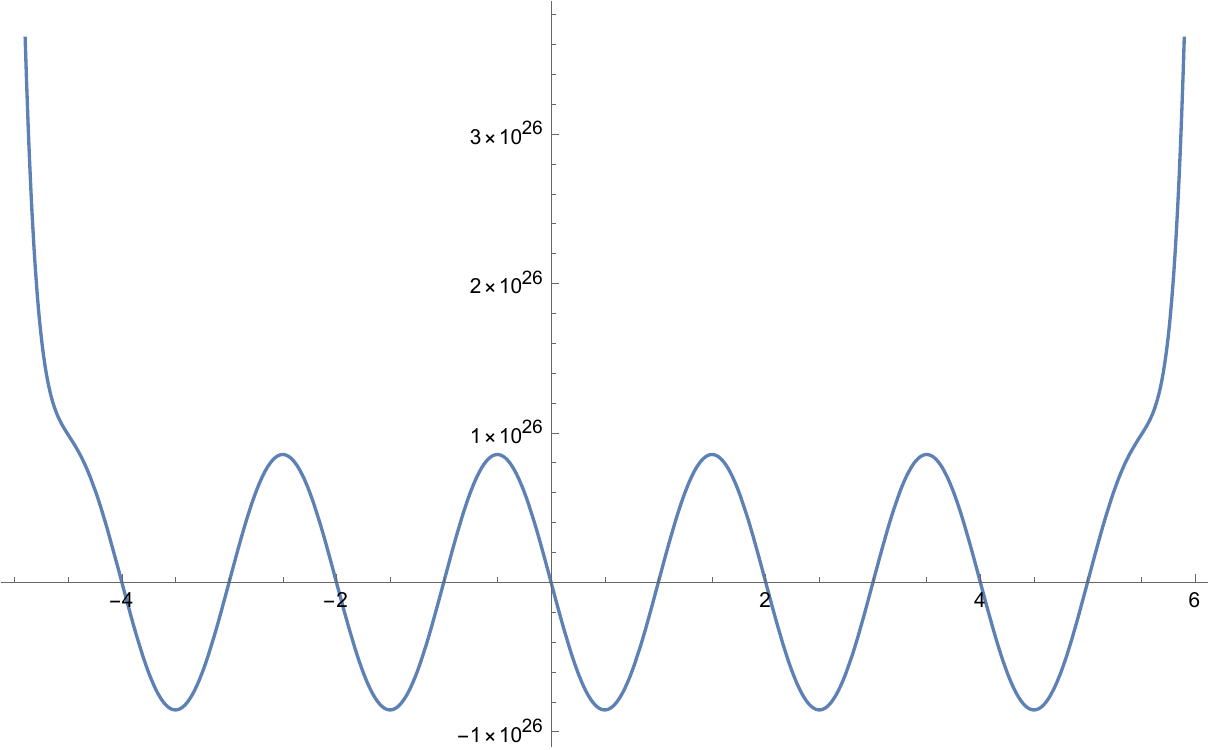}
\caption{$E_{38}(x)$}
\end{subfigure}
\quad
\begin{subfigure}{0.45\textwidth}
\includegraphics[width=\textwidth]{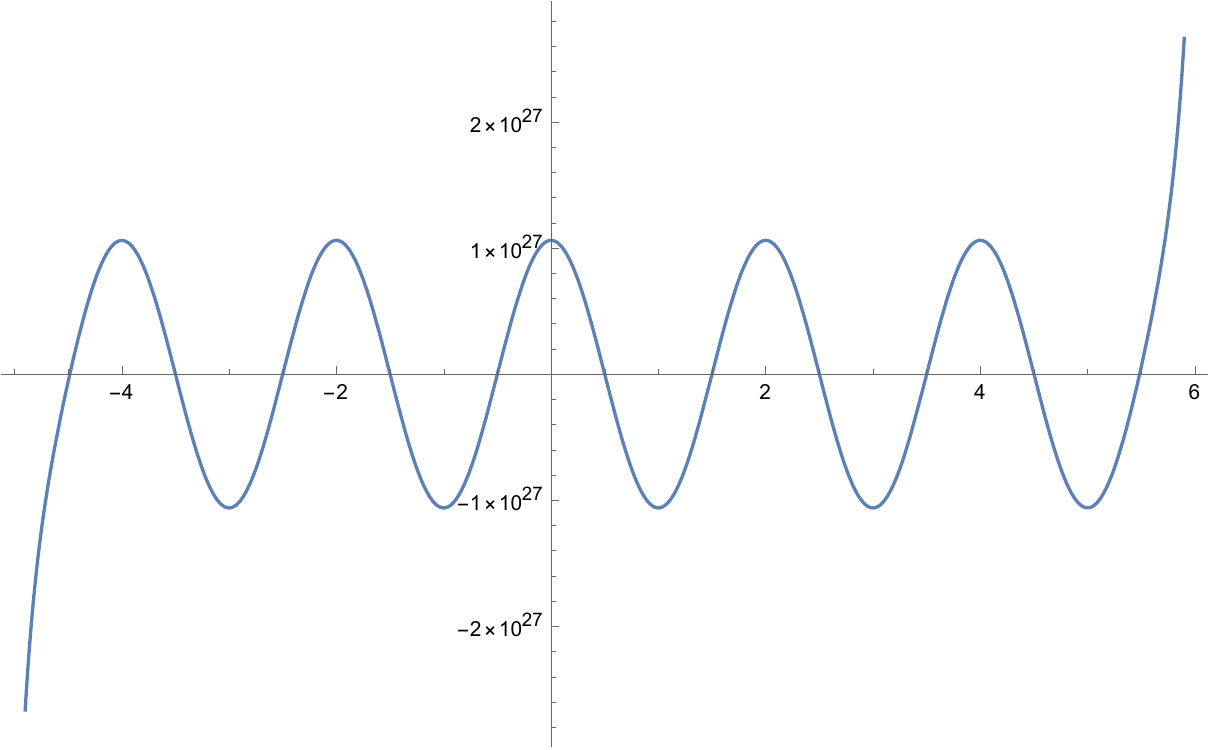}
\caption{$E_{39}(x)$}
\end{subfigure}
\caption{Plots of Euler polynomials $E_n(x)$ for real $x$}
\label{fig:eulerplots}
\end{figure}

 Further, the distributions of the zeroes of $\Phi_n(z)$ in the complex plane (see Figure \ref{fig:zeroplots}) appear similar to those of $B_n(x)$ and $E_n(x)$ (see Figure \ref{fig:bernoullizeroes}), with the largest real zero of $\Phi_n(z)$ appearing to grow approximately linearly with $n.$

\begin{figure}[htb]
\centering
\begin{subfigure}{0.43\textwidth}
\includegraphics[width=\textwidth]{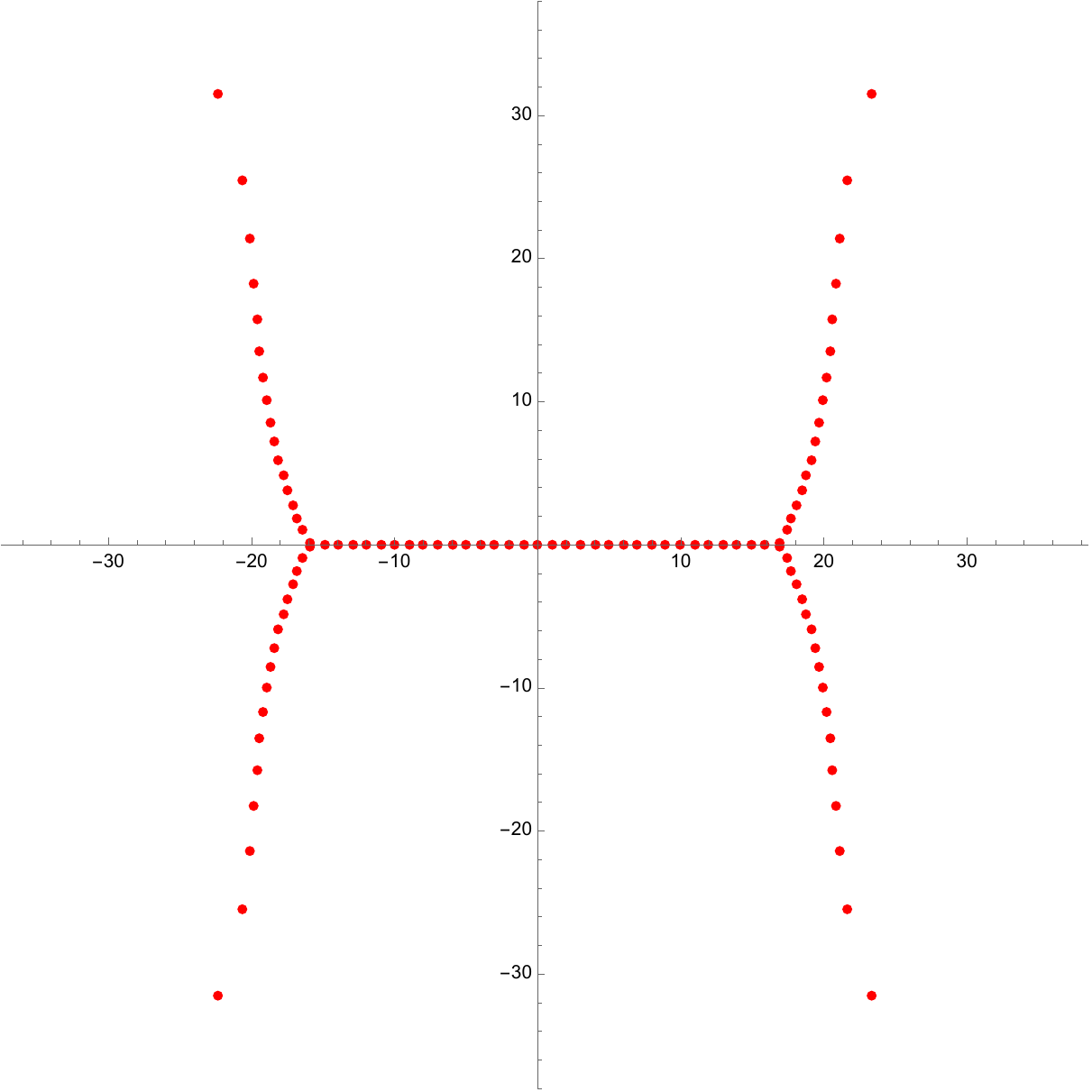}
\caption{$\Phi_{100}(z)$}
\end{subfigure}
\quad
\begin{subfigure}{0.43\textwidth}
\includegraphics[width=\textwidth]{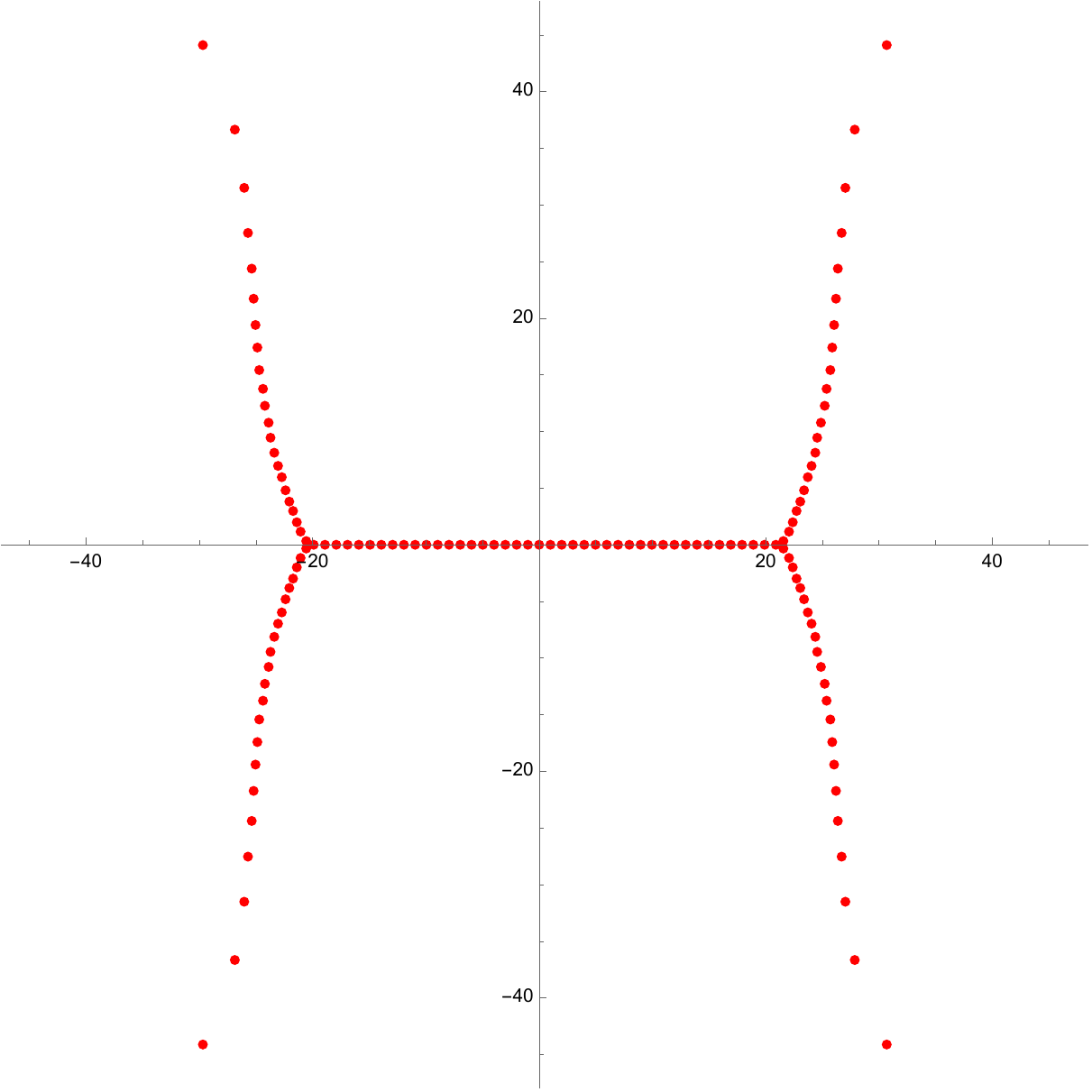}
\caption{$\Phi_{130}(z)$}
\end{subfigure}
\caption{Zeroes of $\Phi_n(z)$}
\label{fig:zeroplots}
\end{figure}

The difference is that the zeroes of largest magnitude of $\Phi_n(z)$ do not appear to approach the imaginary axis as $n$ increases, in contrast to the case of $B_n(x)$ and $E_n(x)$ in which the limiting shape of the zeroes as $n \rightarrow \infty$, under an appropriate scaling, is a closed curve \cite{boyergoh, dilcher2} related to that obtained by Szeg\"o for the zeroes of Taylor polynomials of the exponential function \cite{szego}.

\begin{figure}[h]
\begin{subfigure}{0.43\textwidth}
\includegraphics[width=\textwidth]{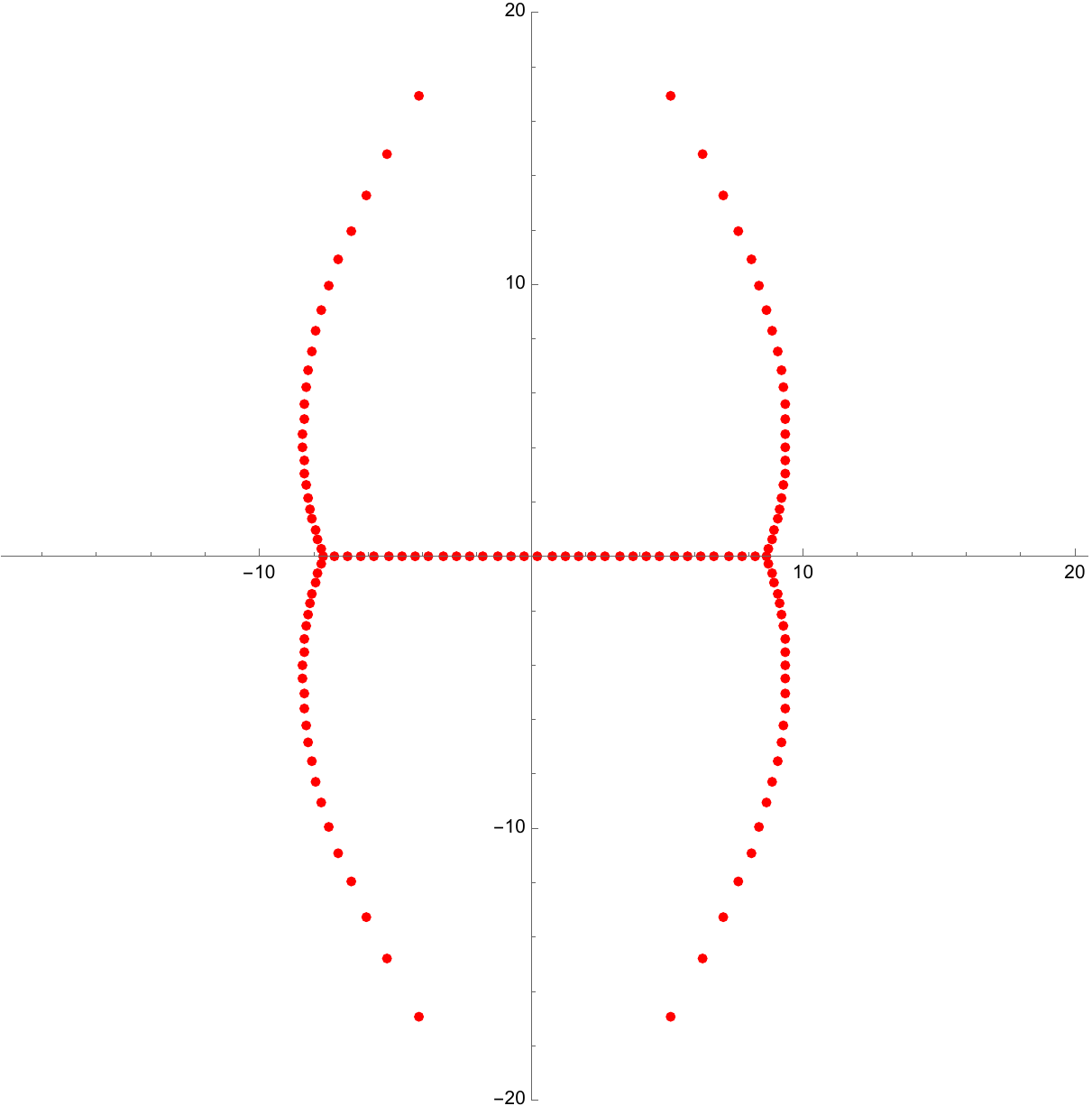}
\caption{$B_{130}(x)$}
\end{subfigure}
\quad
\begin{subfigure}{0.43\textwidth}
\includegraphics[width=\textwidth]{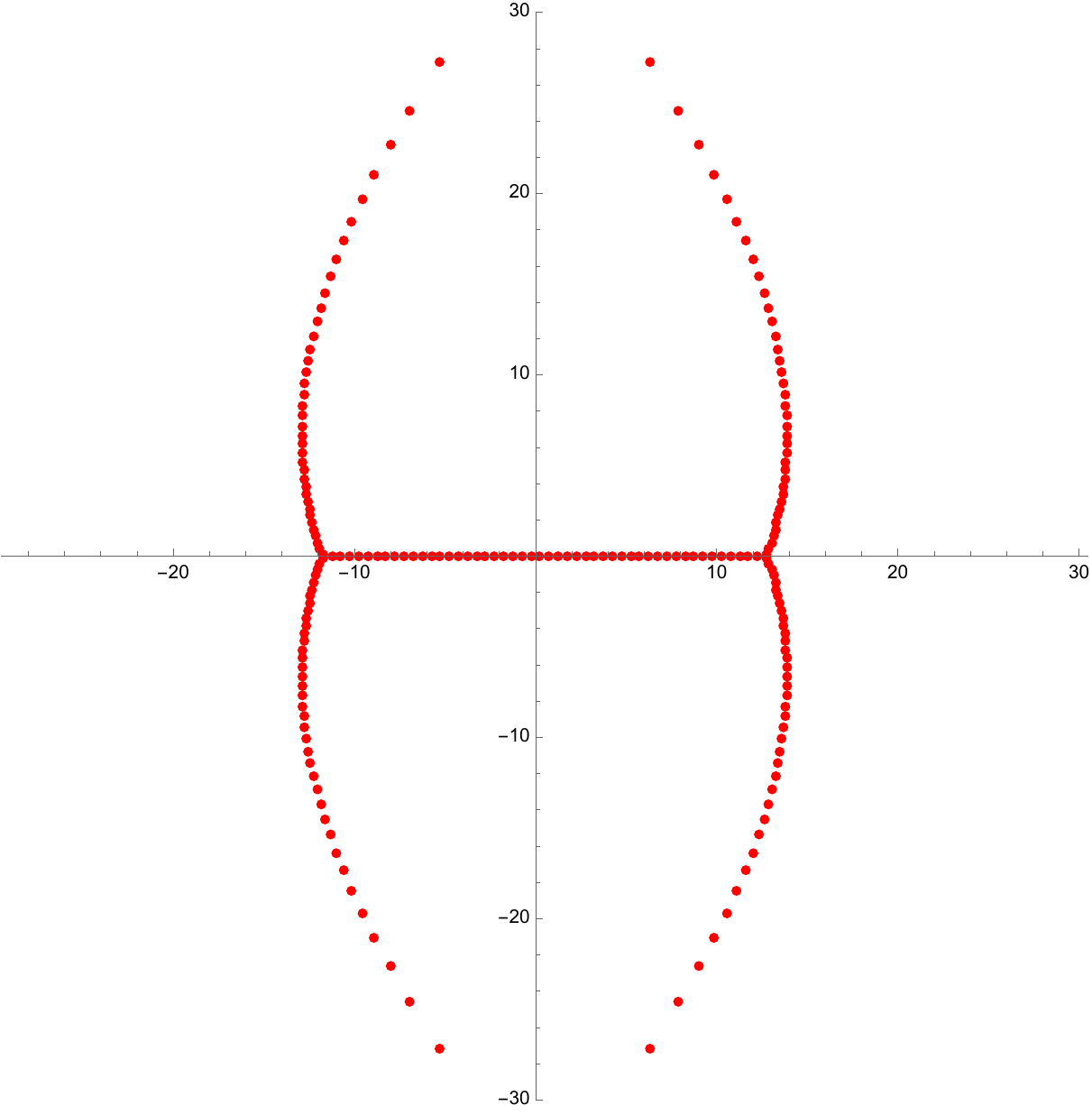}
\caption{$B_{200}(x)$}
\end{subfigure}
\caption{Zeroes of $B_n(x)$}
\label{fig:bernoullizeroes}
\end{figure}


The number of real zeroes of the polynomial $\Phi_n(z)$ also appears to grow approximately linearly with $n$, 
which is similar to $B_n(x)$ (see \cite{inkeri, delange1, delange2, leeming, veselovward, edwardsleeming}).


It may be instructive, due to the operator $2 \partial^{-1} \tanh(\partial / 2)$ appearing in the defining recurrence for $\Phi_n(z;\tau)$, to also look at the version of the Euler polynomials defined by
\begin{equation} \label{eulerishdef}
\widetilde{E}_n (z): = 2 \partial^{-1} \tanh(\partial / 2) z^n = \partial^{-1} (T-1) E_n(z), \qquad \widetilde{E}_n(0) = 0,
\end{equation}
The first few of these are 
$$
\widetilde{E}_1(z) = z, \,\,\widetilde{E}_2(z)=z^2, \,\, \widetilde{E}_3(z)= z^3 - \frac{1}{2} z, \,\, \widetilde{E}_4(z)= z^4 - z^2,\,\, \widetilde{E}_5(z) = z^5 - \frac{5}{3} z^3 + z, 
$$
$$
\widetilde{E}_6(z) = z^6 - \frac{5}{2} z^4 + 3 z^2, \,\, \widetilde{E}_7(z)= z^7 - \frac{7}{2} z^5 + 7 z^3 - \frac{17}{4} z,\,\,
\widetilde{E}_8(z)= z^8 - \frac{14}{3} z^6 + 14 z^4 - 17 z^2. 
$$
These polynomials satisfy the symmetry relation
\begin{equation}
\widetilde{E}_n(-z) = (-1)^n \widetilde{E}_n(z),
\end{equation}
which follows directly from the definition \eqref{eulerishdef}.

\subsection{Outer expansion}

The plots in Figure \ref{fig:phiplots2} suggest that the polynomials exhibit oscillatory behaviour in the `bulk' between the largest and smallest real zeroes, before transitioning to power law-type behaviour when the leading power dominates for large real $z$. 

To investigate the behaviour of the polynomials in this outer region we return to the delay Painlev\'e-I  equation in the form 
\begin{equation}
(T+1) f ' = (T-1) f^2+2\beta, \qquad T\varphi(z)=\varphi(z+1),
\end{equation}
and introduce the scaled independent variable $\xi$ according to
\begin{equation}
\xi = \beta(z - \tfrac{1}{2}).
\end{equation}
The shift in $z$ acts on $u(\xi) = f(z)$ by $T_\beta \, u (\xi) = u(\xi + \beta)$, and the equation becomes 
\begin{equation} \label{scaled}
\beta (T_\beta+1) u'  =  (T_\beta-1) u^2 +2\beta.
\end{equation}
This equation is satisfied if
\begin{equation} \label{new1}
u^2=\frac{e^{\beta D}+1} {e^{\beta D}-1}\beta D u-2\xi-\sum_{k=1}^\infty \tilde s_k \beta^k,
\end{equation}
where $D=\frac{d}{d\xi}$ and $\tilde s_k, \,\, k \in \mathbb N$ are the integration constants. Indeed, applying to both sides of this equation the operator $T_\beta-1=e^{\beta D}-1$, we get equation \eqref{scaled}.


\begin{lemma}
The integration constants $\tilde s_k=s_k$ are the same as in the conserved quantity (\ref{cons3}),(\ref{sk}). 
\end{lemma}

\begin{proof} To see this apply to the equation \eqref{new1} the operator $(e^{\beta D}-1)(\beta D)^{-1}$ acting as
$$
(e^{\beta D}-1)(\beta D)^{-1}\phi(\xi)=\frac{1}{\beta}\int_\xi^{\xi+\beta}\phi(t)dt:
$$
$$
\frac{1}{\beta}\int_\xi^{\xi+\beta}u^2(t)dt=(e^{\beta D}+1) u(\xi)-\frac{1}{\beta}\int_\xi^{\xi+\beta}(2t+\sum_{i=1}^\infty \tilde s_i \beta^i)dt=(T_\beta+1) u-2\xi-\beta-\sum_{i=1}^\infty \tilde s_i \beta^i.
$$
Comparing this with the conserved quantity $S$ given by (\ref{cons3}) and taking into account the change $\xi=\beta(z-\tfrac{1}{2})$ we have the claim.
\end{proof} 
Using now the Taylor expansion
$$\frac{t}{2}\frac{e^t+1}{e^t-1}=\frac{t}{2}\coth \frac{t}{2}=\sum_{n=1}^{\infty}\frac{B_{2n}}{(2n)!}t^{2n}=1+\frac{1}{12}t^2-\frac{1}{720}t^4+\dots,$$
where $B_{k}$ are Bernoulli numbers, we can rewrite the equation \eqref{new1} as
\begin{equation} \label{new2}
u^2= 2\mathcal L u-2\xi-\sum_{i=1}^\infty s_i \beta^i,
\end{equation}
where
\begin{equation} \label{newL}
\mathcal L=\sum_{n=0}^{\infty}\frac{B_{2n}}{(2n)!}\beta^{2n}D^{2n}=1+\frac{1}{12}\beta^2D^2-\frac{1}{720}\beta^4D^4+\dots.
\end{equation}

We look for a solution as a series in $\beta$:
\begin{equation} \label{seriesnew}
u(\xi) = \sum_{k=0}^{\infty} u_{k} (\xi) \beta^{k}.
\end{equation}
For $k=0$ we have 
$u_0^2=2u_0-2\xi,$ implying that
$
u_0=1-\sqrt{1-2\xi}.
$
Continuing this we have
$$
u_1=\frac{s_1}{2(1-u_0)}=\frac{s_1}{2\sqrt{1-2\xi}}, \quad u_2=\frac{u_1^2-\frac{1}{6}u_0''+s_2}{2(1-u_0)}=\frac{s_1^2}{8}(1-2\xi)^{-3/2}-\frac{1}{12}(1-2\xi)^{-2}+\frac{s_2}{2}(1-2\xi)^{-1/2}.
$$
Let now $w:=z-\frac{1}{2},$ so we have $\xi=\beta w$ and the expansion
$$
u=-\sum_{n=1}^\infty (-2)^n\binom{1/2}{n}\beta^n w^n+ \beta \frac{s_1}{2} \sum_{n=0}^\infty (-2)^{n}\binom{-1/2}{n}\beta^n w^n
$$
$$
+ \beta^2 \left[\frac{s_1^2}{8}\sum_{n=0}^\infty (-2)^{n}\binom{-3/2}{n}\beta^n w^n-\frac{1}{12}\sum_{n=0}^\infty (-2)^n\binom{-2}{n}\beta^n w^n
+\frac{s_2}{2}\sum_{n=0}^\infty (-2)^{n}\binom{-1/2}{n}\beta^n w^n\right]+\dots
$$
The coefficients in the expansion
$$
u=\sum_{n=0}^{\infty} \tilde\Phi_{n} (w;s) \beta^{n}
$$
are polynomials in $w$:
\begin{equation} \label{tildephi}
\tilde\Phi_n=-(-2)^n\binom{1/2}{n}w^n-(-2)^{n-2}s_1\binom{-1/2}{n-1} w^{n-1}
\end{equation}
$$
+(-2)^{n-2}\left[\frac{s_1^2}{8}\binom{-3/2}{n-2}-\frac{1}{12} \binom{-2}{n-2}+\frac{s_2}{2}\binom{-1/2}{n-2}\right]w^{n-2}+\dots.
$$

In particular, the first few of these are
\begin{equation*}
\begin{aligned}
\tilde\Phi_0&=0, \\
\tilde\Phi_1&=w+\tfrac{1}{2}s_1,\\
\tilde\Phi_2&=\tfrac{1}{2}w^2+\tfrac{1}{2}s_1 w+\tfrac{1}{8}s_1^2+\tfrac{1}{2}s_2-\tfrac{1}{12},\\
\tilde\Phi_3&=\tfrac{1}{3} w^3 + \tfrac{3}{4} s_1 w^2 + \left(\tfrac{3}{8}s_1^2 + \tfrac{1}{2}s_2  - \tfrac{1}{3}\right)w + \tfrac{1}{16}s_1^3 + \tfrac{1}{4} s_1 s_2 - \tfrac{1}{6}s_1 + \tfrac{1}{2}s_3, \\
\tilde\Phi_4&= \tfrac{5}{8} w^4 + \tfrac{5}{4} s_1 w^3 + \left(\tfrac{15}{16}s_1^2 + \tfrac{3}{4}s_2- 1 \right)w^2 + \left(\tfrac{5}{16}s_1^3 + \tfrac{3}{4} s_1 s_2 - s_1 + \tfrac{1}{2}s_3\right) w \\
&\quad + \tfrac{5}{128}s_1^4 - \tfrac{1}{4}s_1^2 + \tfrac{3}{16} s_1^2 s_2 + \tfrac{1}{4}s_1 s_3 + \tfrac{1}{8} s_2^2 - \tfrac{1}{6} s_2 + \tfrac{1}{2} s_4 + \tfrac{55}{288}, \\
\tilde\Phi_5&= \tfrac{7}{8} w^5 + \tfrac{35}{16}s_1 w^4 + \left( \tfrac{35}{16}s_1^2 + \tfrac{5}{4}s_2- \tfrac{8}{3}\right) w^3 + \left(\tfrac{35}{32}s_1^3 - 4 s_1+ \tfrac{15}{8} s_1 s_2+\tfrac{3}{4}s_3\right) w^2 \\
&\quad + \left( \tfrac{35}{128}s_1^4 -2s_1^2 + \tfrac{15}{16}s_1^2s_2 + \tfrac{3}{4}s_1s_3 + \tfrac{3}{8}s_2^2 - s_2 +\tfrac{1}{2}s_4 + \tfrac{163}{96}\right)w \\
&\quad + \tfrac{7}{256} s_1^5 -\tfrac{1}{3}s_1^3 + \tfrac{163}{192}s_1 + \tfrac{5}{32}s_1^3 s_2 + \tfrac{3}{16}s_1s_2^2  + \tfrac{3}{16} s_1^2 s_3 +\tfrac{1}{4}s_2s_3 -\tfrac{1}{6}s_3 + \tfrac{1}{4}s_1s_4 + \tfrac{1}{2}s_5.
\end{aligned}
\end{equation*}


\begin{proposition}
The polynomials $\tilde \Phi_n(w;s)$  coincide with the polynomials $\Phi_n(z;\tau)$ from the previous section
\begin{equation} \label{reltil}
\tilde \Phi_n(w;s)=\Phi_n(z; \tau),
\end{equation}
after the substitution $w=z-\tfrac{1}{2}$ and the triangular change of parameters:
\begin{equation} \label{change}
\tau_n=s_n+ T_n(s_1,\dots, s_{n-1}), \quad T_n\in \mathbb Q[s_1,\dots, s_{n-1}],
\end{equation}
which is the inversion of the transformation (\ref{change2}).

If all $s_{2k-1}=0$, then the same is true for $\tau_{2k-1}$ and the corresponding polynomials have the parity property
\begin{equation}
\tilde \Phi_n(-w;s)=(-1)^n \tilde \Phi_n(w;s).
\end{equation}
\end{proposition}

Explicitly for $k\leq 6$ the change of parameters has the following form:

\begin{equation} 
\begin{aligned}
\tau_1&=s_1, \\
\tau_2&=s_2+\tfrac{s_1^2}{4}+\tfrac{1}{12}, \\
\tau_3&=s_3+\tfrac{s_1^3}{8}+\tfrac{s_1 s_2}{2}+\tfrac{s_1}{24}, \\
\tau_4&= s_4+\tfrac{5 s_1^4}{64}+\tfrac{3s_1^2s_2 }{8} -\tfrac{s_1^2}{32}+\tfrac{s_1s_3}{2}+\tfrac{s_2^2}{4}+\tfrac{s_2}{24}-\tfrac{23}{576},\\
\tau_5&= s_5 + \tfrac{7s_1^5}{128}  + \tfrac{5s_1^3s_2 }{16} -\tfrac{23 s_1^3}{192}+\tfrac{3 s_1^2s_3}{8} +\tfrac{3 s_1s_2^2}{8} -\tfrac{s_1 s_2}{16}+\tfrac{s_1 s_4}{2}-\tfrac{11 s_1}{384}+\tfrac{s_2 s_3}{2}+\tfrac{s_3}{24},\\
\tau_6&= s_6 + \tfrac{21 s_1^6}{512}+\tfrac{35s_1^4s_2 }{128} -\tfrac{335 s_1^4}{1536}+\tfrac{5s_1^3 s_3 }{16} +\tfrac{15 s_1^2s_2^2}{32} -\tfrac{23 s_1^2s_2}{64} +\tfrac{3s_1^2s_4}{8} +\tfrac{373 s_1^2}{1536}+\tfrac{3s_1s_2 s_3 }{4}\\
&\qquad  +\tfrac{s_1 s_5}{2}-\tfrac{s_1 s_3}{16}+\tfrac{s_2^3}{8}+\tfrac{s_3^2}{4}+\tfrac{s_2 s_4}{2}+\tfrac{s_4}{24}-\tfrac{s_2^2}{32}-\tfrac{11 s_2}{384}+\tfrac{2923}{13824}.\\
\end{aligned}
\end{equation}

In the special case when all parameters $s_i=0$ we have the polynomials
\begin{equation} \label{tildeQ}
\tilde Q_n(w):=\tilde \Phi_n(w;0):
\end{equation}
$$
\tilde Q_0=0, \quad
\tilde Q_1=w, \quad \tilde Q_2=\tfrac{1}{2}w^2-\tfrac{1}{12}, \quad \tilde Q_3=\tfrac{1}{3}w^2-\tfrac{1}{3}w,
$$
$$
 \tilde Q_4=\tfrac{5}{8}w^4-w^2+\tfrac{55}{288}, \quad \tilde Q_5=\tfrac{7}{8}w^5-\tfrac{8}{3}w^3+\tfrac{163}{96}w.
$$

Let $a^{(\nu)}_j$ and $b^{(\nu)}_j$ be the coefficients of these polynomials for even and odd $n$ respectively:
\begin{equation}
Q_{2\nu} = \sum_{j=0}^{\nu} a_{j}^{(\nu)} w^{2j}, \qquad Q_{2\nu-1} = \sum_{j=1}^{\nu} b_{j}^{(\nu)} w^{2j-1}.
\end{equation}

\begin{proposition} \label{prop:alternatingsigns}
All the coefficients $a^{(\nu)}_j$ and $b^{(\nu)}_j$ are non-zero and have alternating signs:
\begin{equation}
\operatorname{sgn} a_{j}^{(\nu)} = (-1)^{\nu-j}, \qquad \operatorname{sgn} b_{j}^{(\nu)} = (-1)^{\nu-j}. 
\end{equation}
\end{proposition}
\begin{proof}
We proceed by induction. Using Proposition \ref{prop:5.1} from the next section, through a long but direct computation we obtain the following recursive formulas for the coefficients:
\begin{equation}
\begin{aligned}
b_t^{(n)} &= \sum_{r=1}^{n-t+1} \frac{2(2^{2r}-1)}{(2r)!} \frac{(2t+2r-3)!}{(2t-1)!} B_{2r}  \sum_{\mu=1}^{n-1} 
\sum_{j+k=t+r-1} 
\left( a_j^{(\mu)} b_{k}^{(n-\mu)} + b_{j}^{(\mu)} a_k^{(n-\mu)}\right) \\
a_t^{(n)} &= \sum_{r=1}^{n-t+1} \frac{2(2^{2r}-1)}{(2r)!} \frac{(2t+2r-2)!}{(2t)!} B_{2r} \sum_{\mu=1}^{n-1} 
{\Bigg[} 
\sum_{j+k=t+r-1}  a_j^{(\mu)} a_{k}^{(n-\mu)} 
+\sum_{j'+k'=t+r}  b_{j'}^{(\mu)} b_{k'}^{(n-\mu-1)},
{\Bigg]}  
\end{aligned}
\end{equation}
where $B_k$ are the Bernoulli numbers.
Then assume $a_j^{(\nu)} = (-1)^{\nu-j} |a_j^{(\nu)}|$ and $b_j^{(\nu)} = (-1)^{\nu-j} |b_j^{(\nu)}|$ for $\nu < n$, and note the alternating sign property $B_{2r}=(-1)^{r+1} | B_{2r}|$ of the Bernoulli numbers. 
The first formula above leads to 
\begin{equation}
b_t^{(n)} = (-1)^{n-t}\sum_{r=1}^{n-t+1} \frac{2(2^{2r}-1)}{(2r)!} \frac{(2t+2r-3)!}{(2t-1)!} |B_{2r}| \sum_{\mu=1}^{n-1} 
\sum_{j+k=t+r-1} 
\left( |a_j^{(\mu)} b_{k}^{(n-\mu)}| + |b_{j}^{(\mu)} a_k^{(n-\mu)}|\right),
\end{equation}
so $(-1)^{n-t} b_{t}^{(n)} > 0$ as required.
 The fact that $(-1)^{n-t} a_{t}^{(n)} > 0$ follows similarly from the second formula.
\end{proof}

Note that the same property holds for the shifted Bernoulli polynomials $\tilde B_n(x):=B_n(x+\frac{1}{2})$.
Indeed, they can be expressed in terms of the Bernoulli numbers as
\begin{equation} \label{tildeB}
\tilde B_n(x)=\sum_{k=0}^n\binom{n}{k}(2^{k-n+1}-1)B_{n-k}x^k,
\end{equation}
which follows from the well-known formulas
$$
B_n(x+a)=\sum_{k=0}^n\binom{n}{k}B_{n-k}(a)x^k, \quad
B_n(\tfrac{1}{2})=(2^{1-n}-1)B_n.
$$
We continue the parallel with the Bernoulli polynomials in the next section.


\section{Special case $s_i=0$: Bernoulli-Catalan polynomials} \label{sec:alls_izero}

When all parameters $s_i=0$ the equation (\ref{new2}) takes the form
\begin{equation} \label{new3}
u^2= 2\mathcal L u - 2\xi.
\end{equation}
This leads to the following recurrence relation for the corresponding polynomials 
\begin{equation} \label{Qn}
Q_n(z):=\tilde Q_n(z-\tfrac{1}{2})=\tilde\Phi_n(z-\tfrac{1}{2}; 0).
\end{equation}
\begin{proposition} \label{prop:5.1}
The polynomials $Q_n(z)$
can be determined by the recursive relation
\begin{equation}
\label{Qnrel}
Q_n(z)=\mathcal K \sum_{k=1}^{n-1}Q_k(z)Q_{n-k}(z),\quad n>1, \quad Q_0(z):=0, \,\,  Q_1(z):=z-\frac{1}{2},
\end{equation}
where
$$
\mathcal K= K(\tfrac{d}{dz}), \quad K(t)=\frac{1}{t}\tanh\frac{t}{2}.
$$
\end{proposition}

The Taylor expansion of $\tanh x$ is known to be
$$
\tanh x=\sum_{n=0}^\infty\frac{2^{2n}(2^{2n}-1)}{(2n)!}B_{2n}x^{2n-1},
$$
where $B_{2n}$ are Bernoulli numbers, so
\begin{equation}
\label{K}
\mathcal K=2\sum_{n=1}^\infty\frac{2^{2n}-1}{(2n)!}B_{2n}\left(\frac{d}{dz}\right)^{2n-2}=\frac{1}{2}-\frac{1}{24}\left(\frac{d}{dz}\right)^{2}+\frac{1}{240}\left(\frac{d}{dz}\right)^{4}-\frac{17}{40320}\left(\frac{d}{dz}\right)^{6}+\dots
\end{equation}

Since the polynomials $Q_n(z)$ satisfy the Bernoulli modification \eqref{Qnrel} of the Catalan recurrence relation \eqref{recCat}, we call them {\it Bernoulli-Catalan polynomials}.
They satisfy the Bernoulli symmetry relation
$$
Q_n(1-z)=(-1)^n Q_n(z)
$$
and the coefficients of the shifted polynomials $\tilde Q_n(w)=Q_n(w+\frac{1}{2})$ have alternating signs (see Proposition \ref{prop:alternatingsigns}).

Here are the first few of them:
\begin{align*}
Q_0(z) &= 0,\\
Q_1(z) &= z-\tfrac{1}{2}, 							\\
Q_2(z) &= \tfrac{z^2}{2} -\tfrac{z}{2}+\tfrac{1}{24},			\\
Q_3(z) &= \tfrac{z^3}{2}-\tfrac{3 z^2}{4}+\tfrac{z}{24}+\tfrac{5}{48},	\\
Q_4(z) &= \tfrac{5 z^4}{8}-\tfrac{5 z^3}{4}-\tfrac{z^2}{16}+\tfrac{11 z}{16}-\tfrac{23}{1152},	\\
Q_5(z) &=\tfrac{7 z^5}{8}-\tfrac{35 z^4}{16}-\tfrac{23 z^3}{48}+\tfrac{93 z^2}{32}-\tfrac{11 z}{384}-\tfrac{139}{256},\\
Q_6(z) &=\tfrac{21 z^6}{16}-\tfrac{63 z^5}{16}-\tfrac{335 z^4}{192}+\tfrac{965 z^3}{96}+\tfrac{373 z^2}{768}-\tfrac{1579 z}{256}+\tfrac{2923}{27648},\\
Q_7(z) &=\tfrac{33 z^7}{16}-\tfrac{231 z^6}{32}-\tfrac{331 z^5}{64}+\tfrac{3965 z^4}{128}+\tfrac{1111 z^3}{256}-\tfrac{21041 z^2}{512}+\tfrac{4289 z}{27648}+\tfrac{441049}{55296},\\
Q_8(z) &=\tfrac{429 z^8}{128}-\tfrac{429 z^7}{32}-\tfrac{5327 z^6}{384}+\tfrac{11333 z^5}{128}+\tfrac{71225 z^4}{3072}-\tfrac{107303 z^3}{512}-\tfrac{417451 z^2}{55296}+\tfrac{7151341 z}{55296}-\tfrac{20794951}{13271040}.
\end{align*}

Let $A_k^{(n)}$ be the coefficients of $Q_n(z)$:
$$
Q_{n}(z)=\sum_{k=0}^n A_k^{(n)} z^{k}.
$$
As a corollary of Proposition \ref{prop:leadingcoeffs} we have the following 

\begin{proposition}
\label{leadq}
The leading coefficients of the Bernoulli-Catalan polynomials $Q_n(z), \, n\geq 1$ are
\begin{equation}
\label{leadQ}
A_{n}^{(n)}=A_n, \quad A_{n-1}^{(n)}=-\frac{n}{2}A_n,\quad A_{n-2}^{(n)} =\frac{n(n-1)}{8}A_n-\frac{2^{n-2}(n-1)}{12}.
\end{equation}
where, as before, $A_n=2^{1-n}C_{n-1}=\frac{(2n-3)!!}{n!}$ and $C_n$ are the Catalan numbers.
\end{proposition}

The Bernoulli-Catalan polynomials have the form
$$Q_n(z)=\Phi_n(z; \tau^*),$$
where $\tau_k^*$ are the special values of parameters $\tau_k$ corresponding to the zero values of $s_i$, which we call the {\it Bernoulli-Catalan numbers}:
\begin{equation} \label{special}
\tau_1^*=0, \,\, \tau_2^*=\frac{1}{12},\,\, \tau_3^*=0,\,\, \tau_4^*=-\frac{23}{576},\,\, \tau_5^*=0\,\, \tau_6^*=\frac{2923}{13824},\,\, \tau_7^*=0,\,\, \tau_8^*=-\frac{20794951}{6635520}, \dots
\end{equation}
We conjecture that, like in the Bernoulli case, the signs of $\tau_{2n}^*$ are alternating.

Note that using the conservation form (\ref{conserved2}), 
we can express the numbers $\tau^*$ as the integrals
\begin{equation} \label{tauint}
\tau_n^* = \sum_{m=1}^{n}  \int_0^1 Q_m(z) Q_{n-m}(z) dz,
\end{equation}
which is reminiscent of both the Catalan recurrence \eqref{recCat} and the relation for the Bernoulli numbers
$$
B_{m+n}=(-1)^m\binom{m+n}{n} \int_0^1 B_m(z)B_n(z)dz,
$$
(see e.g. \cite{MV}, formula (4.12)).


The polynomials $Q_n(z)$ are qualitatively similar to those with $\tau=0$ studied in Section \ref{sec:polynomials}, in particular their behaviour for real $z$ and the distributions of their zeroes in the complex plane (see Figure \ref{fig:Qpolycomparison}).

\begin{figure}[h]
\centering
\begin{subfigure}{0.33\textwidth}
\includegraphics[width=\textwidth]{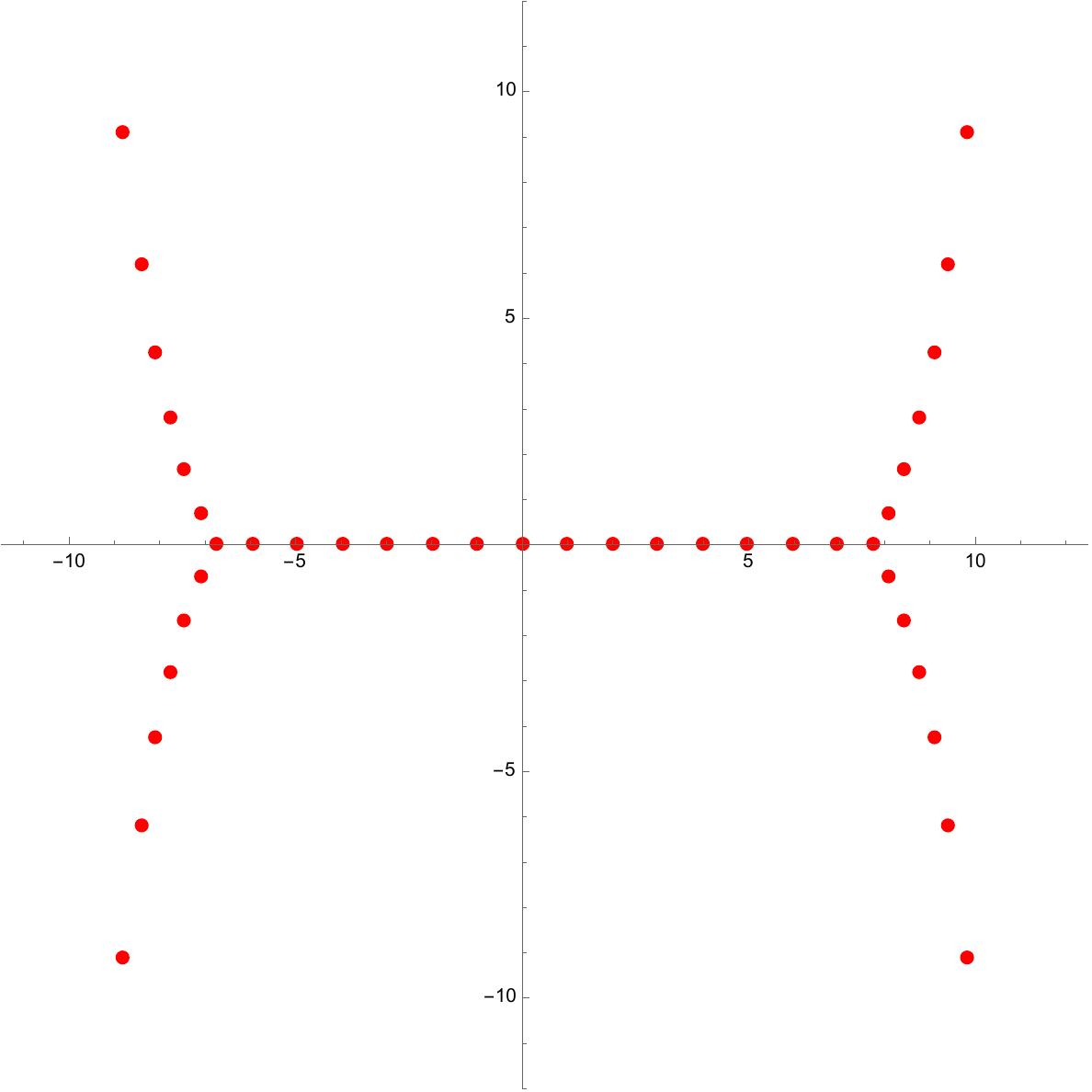}
\caption{Zeroes of $Q_{40}(z)$}
\end{subfigure}
\quad
\begin{subfigure}{0.33\textwidth}
\includegraphics[width=\textwidth]{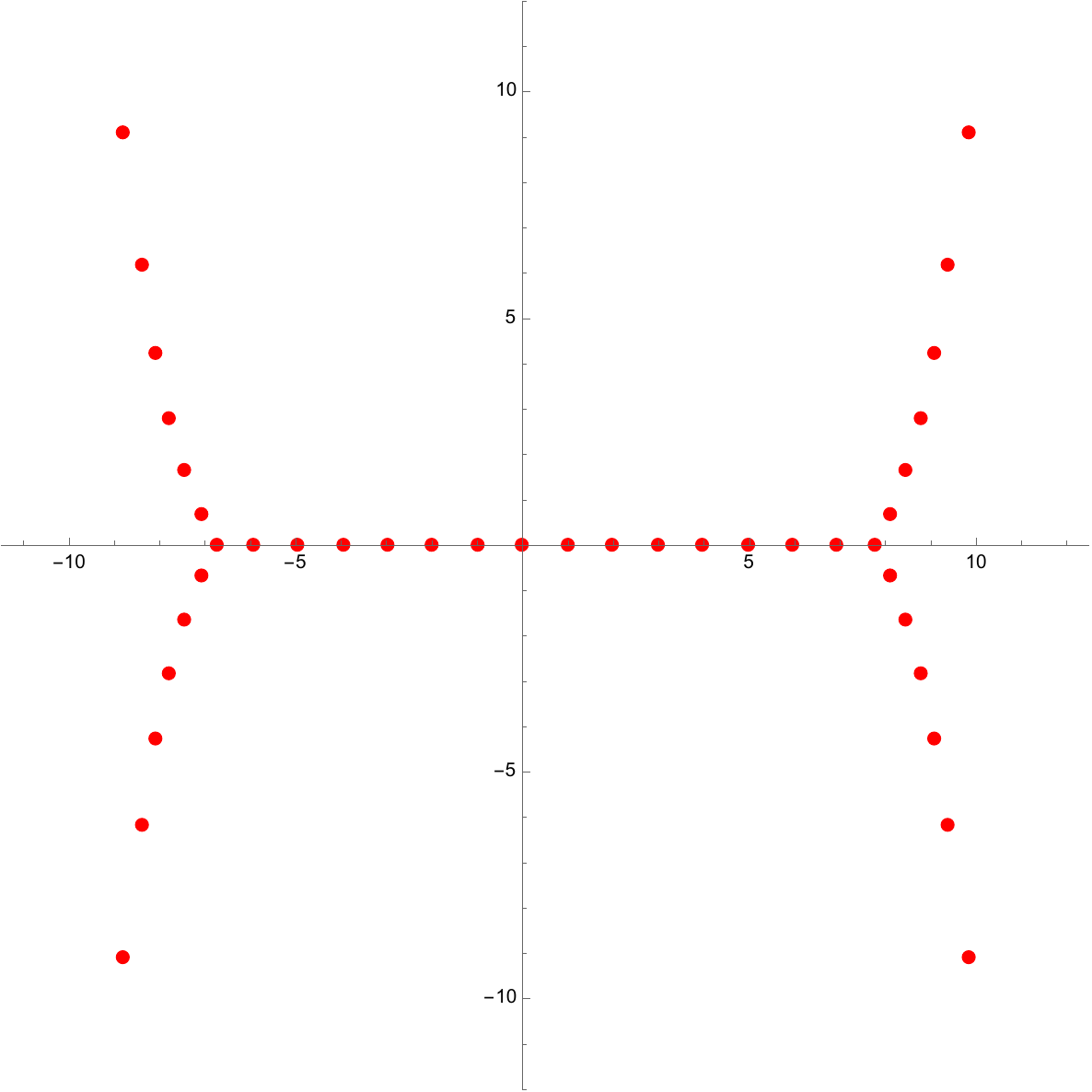}
\caption{Zeroes of $\Phi_{40}(z)$}
\end{subfigure}
\\
\begin{subfigure}{0.33\textwidth}
\includegraphics[width=\textwidth]{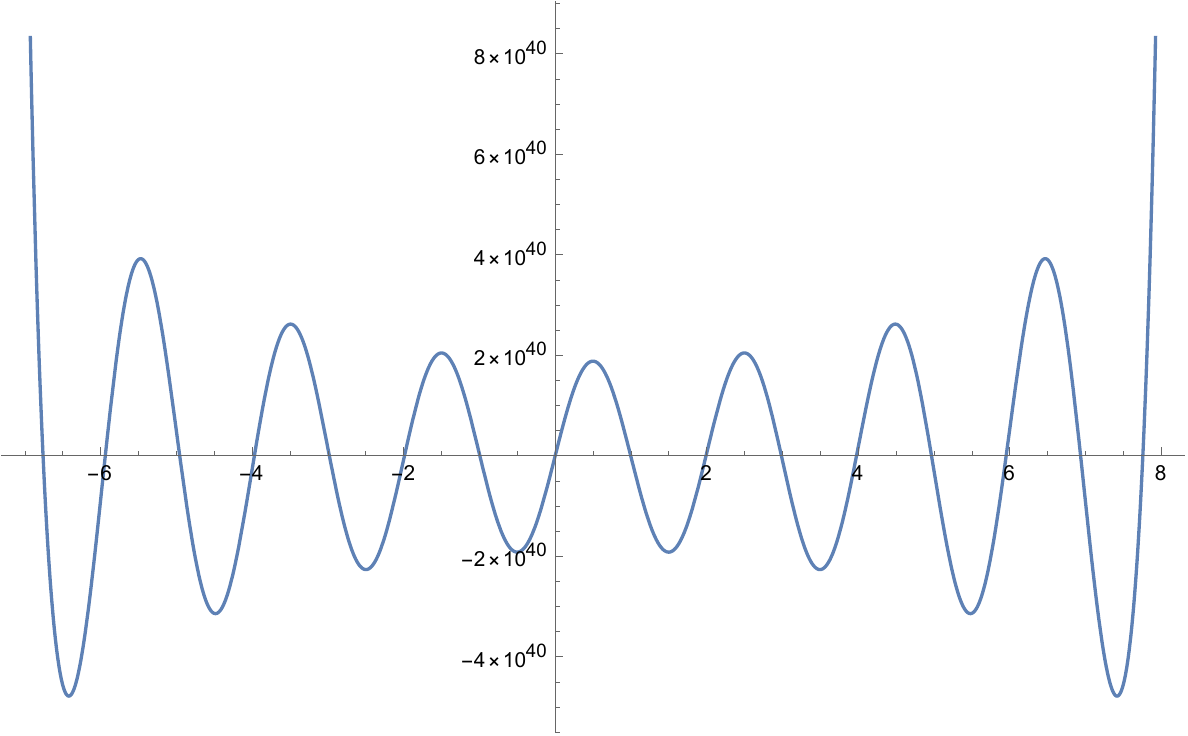}
\caption{$Q_{40}(z)$ for real $z$}
\end{subfigure}
\quad
\begin{subfigure}{0.33\textwidth}
\includegraphics[width=\textwidth]{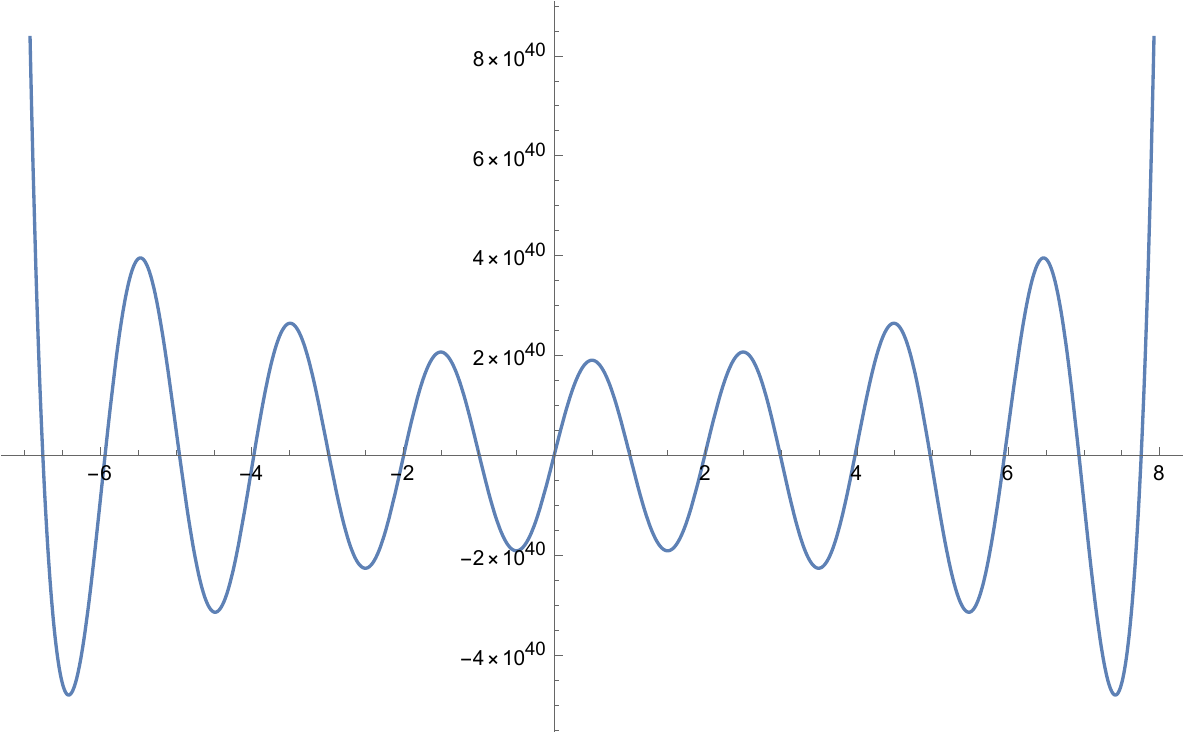}
\caption{$\Phi_{40}(z)$ for real $z$}
\end{subfigure}
\caption{Comparison between $Q_n(z)$ and $\Phi_n(z)$}
\label{fig:Qpolycomparison}
\end{figure}


We will see in Section \ref{sec:masurveech} that the polynomials $Q_n(z)$ naturally appear in relation with the computation of the Masur-Veech volumes.

\section{Meromorphic solutions and singularity structure} \label{sec:meromorphicsols}

We now turn our attention to meromorphic solutions of equation \eqref{delayP1}. 
In the case $\alpha = 0$, Berntson \cite{bjornthesis, Ber18} showed the existence of multiparameter families of elliptic function solutions to this equation. 
We also we have a wealth of meromorphic solutions coming from the class of `anti-periodic' functions \
\begin{equation}
\tilde{V}_0^- = \left\{ f \text{ meromorphic} ~:~ (T+1)f = 0 \right\}.
\end{equation}
For example, it is quick to verify that the Jacobi elliptic function 
\begin{equation}
f(z) = \operatorname{sn}\left( 2K z + z_0; m \right)
\end{equation}
is a solution, where $K$ is the first elliptic integral for the modulus $m$. 
However some of Berntson's solutions come from outside the antiperiodic class. 
For example the function
\begin{equation} \label{zetasol}
f(z)= \zeta(z+a-h ; g_2, g_3) - \zeta(z+a ; g_2, g_3) + \zeta(h ; g_2, g_3)
\end{equation}
can be verified to solve the $\alpha = 0$ case of equation \eqref{delayP1} with shift $h$ using the identity
\begin{equation}
\left( \zeta(z_1 ; g_2,  g_3)+\zeta(z_2 ; g_2,  g_3)+\zeta(z_3 ; g_2,  g_3) \right)^2 + \zeta'(z_1 ; g_2,  g_3) + \zeta'(z_2 ; g_2,  g_3) + \zeta'(z_3 ; g_2,  g_3) = 0,
\end{equation}
which holds if $z_1 +z_2 + z_3 = 0$.
Berntson also presented an alternative form of this family of solutions in terms of the Jacobi $\operatorname{sn}$ function \cite{bjornthesis}.

\subsection{Singularity confinement}

The delay-differential equation \eqref{delayP1} is a special case of one isolated by Grammaticos, Ramani and Moreira among those of so-called bi-Riccati type using a kind of singularity confinement criterion \cite{GRM93}.
Singularity confinement was introduced as an integrability criterion for discrete systems in \cite{GRP91}, by analogy with the Painlev\'e property for ordinary differential equations. 
In the case of discrete systems equivalent to birational mappings of the plane, the singularity confinement property is now known to be equivalent to the existence of a space of initial conditions, i.e. a family of rational surfaces between which the mappings become isomorphisms \cite{SAKAI2001, takenawageometricapproach, MASE}.
In the case of higher-order equations, i.e. birational mappings in higher dimensions, the notion of a space of initial conditions generalises to a family of rational varieties between which the mappings become pseudo-isomorphisms, i.e. isomorphisms in codimension one \cite{dolgachevortland, takenawatsuda, bedfordkimmatrix, takenawacarstea}.
In the case of delay-differential and differential-difference equations the situation is much more complicated and the structure of singularities of an equation can be much richer \cite{delaySC, carstea}. 
In particular, even when the equation can be iterated in a single-valued way as a discrete system on germs of meromorphic functions, it requires some work to formulate a definition of singularity confinement \cite{delaySC}.

To formally state our results on the singularity structure of the delay Painlev\'e-I  equation we make the following definition, following \cite{bjornthesis}.
\begin{definition}
Consider a delay-differential equation of the form 
\begin{equation} \label{delaydiffeq}
\begin{gathered}
F( f, f', \dots, T f, Tf', \dots, T^{m}f, T^m f', \dots ) = 0, \\
f=f(z), \quad T^k f=f(z+k h), \quad '=\frac{d}{dz},
\end{gathered}
\end{equation}
where $F$ is a differential polynomial in $f, \dots, T^m f$ with coefficients rational functions of $z$.
If a sequence of formal Puiseux series in $\zeta =  z-z_0$, say
\begin{equation}
\begin{gathered}
f = f^{(0)}, \bar{f}=f^{(1)} = T f, \dots, f^{(n)} = T^n f,  \quad n\geq m,\\
f^{(k)}  \in \C(\!(\zeta^{1/\rho_k})\!), \quad \rho_k \in \Z_{> 0}, \quad \zeta = z-z_0,
\end{gathered}
\end{equation}
satisfy formally the equation \eqref{delaydiffeq} (as well as its relevant up-shifts if $n>m$), then we say that $(f, f^{(1)},\dots, f^{(n)})$ is an admissible sequence. 
\end{definition}
We sometimes refer to the series $f^{(1)}, f^{(2)}, \dots$ following some $f$ in an admissible sequence as iterates, by analogy with the purely discrete case where singularity analysis involves iterating the equation or system as a birational map. 
A rough definition of singularity confinement as a property of a delay-differential equation of the form \eqref{delaydiffeq} along the lines of how it is used in \cite{GRM93} would be that any admissible sequence beginning with a Taylor series consists only of Laurent series, and contains finitely many iterates with nonzero singular part before returning to Taylor series.

It is common to refer to the leading behaviours of the series in an admissible sequence as a singularity pattern, and to say that the equations admits such a pattern if there exists a corresponding admissible sequence. 
For example, it can be checked by direct calculation that the delay Painlev\'e-I  equation \eqref{delayP1} admits the singularity pattern
\begin{equation} \label{singpatternsimple}
(\operatorname{rg}, \infty^{1}_{-1}, \infty^1_{+1}, \operatorname{rg}),
\end{equation}
where $\infty^r_{a}$ indicates a Laurent series expansion about a pole of order $r$, with leading term $a \zeta^{-r}$, and $\operatorname{rg}$ indicates a Taylor series, standing for `regular'.
It is common to refer to such singularity patterns beginning with regular iterates that return to regularity after finitely many iterates as `confining', with the `singularity' associated with the solution developing a pole being `confined'.

\subsection{An affine Weyl group structure of singularity patterns}
For the delay Painlev\'e-I  equation \eqref{delayP1} we study the admissible patterns involving only Laurent and Taylor series.
First note that if $f$ is a Taylor series, then if the next iterate $\bar{f}$ is a Laurent series with nonzero singular part then it must be of the form $\bar{f} = \frac{- 1 }{\zeta} + \mathcal{O}(1)$. Similarly if $\bar{f}$ is a Taylor series but the previous iterate is a Laurent series with nonzero singular part then it must be of the form $f = \frac{+1}{\zeta} + \mathcal{O}(1)$.
Further, if a pair of Laurent series 
$f = \sum_{i = - \rho}^{\infty} a_i \zeta^i$,  $\bar{f} = \sum_{i = - \bar{\rho}}^{\infty} \bar{a}_i \zeta^i$
with $\rho, \bar{\rho} > 0$ formally satisfy equation \eqref{delayP1} then $\rho = \bar{\rho}$.
Therefore to study the admissible sequences beginning with a Taylor series it is sufficient to consider pairs of Laurent series 
\begin{equation} \label{laurentseriessimplepoles}
f = \sum_{i = - 1}^{\infty} a_i \zeta^i , \qquad \bar{f} = \sum_{i = - 1}^{\infty} \bar{a}_i \zeta^i.
\end{equation}
We have the following by direct calculation.
\begin{lemma}
If a pair of formal Laurent series \eqref{laurentseriessimplepoles} satisfy the delay Painlev\'e-I  equation \eqref{delayP1}, then the following relations hold among the coefficients $a_i$, $\bar{a}_i$:
\begin{gather}
\left( \bar{a}_{-1} + a_{-1} \right)\left( \bar{a}_{-1} - a_{-1} + 1 \right) = 0,  \label{singeq1}\\
\left( 2 \bar{a}_{-1} - k - 1 \right)\bar{a}_{k+1}  = \left(2 a_{-1} + k +1 \right) a_{k+1} + \sum_{i=0}^{k} \left( \bar{a}_i \bar{a}_{k-i} + a_i a_{k-i} \right), \qquad k\geq 0. \label{singeq2}
\end{gather}
\end{lemma}
The first equation \eqref{singeq1} appears also in the study of the Painleve property of the dressing chain (see formula (31) in Shabat-Veselov \cite{VS}).
This implies that for a given $a_{-1}$, the coefficient $\bar{a}_{-1}$ is determined to be either 
\begin{equation} \label{residuerelations}
\bar{a}_{-1} = - a_{-1}, \quad \text{or} \quad \bar{a}_{-1} = a_{-1}- 1.
\end{equation}
The two maps $a_{-1} \mapsto \bar{a}_{-1}$ in \eqref{residuerelations} are nothing but the usual action of the affine Weyl group $W(A_1^{(1)})$, i.e. 
\begin{equation}
\begin{aligned}
r : \R &\rightarrow \R, \\
a_{-1} &\mapsto - a_{-1},
\end{aligned}
\qquad 
\begin{aligned}
T : \R &\rightarrow \R, \\
a_{-1} &\mapsto a_{-1} - 1.
\end{aligned}
\end{equation}
After choosing either $r$ or $T$ to determine $\bar{a}_{-1}$ in terms of $a_{-1}$, the rest of the coefficients in $\bar{f}$ are determined by \eqref{singeq2}, i.e. $\bar{a}_k$ is given as a rational function of $a_{-1},\dots, a_{k}$ for any $k\geq 0$, so we can extend $r$ and $T$ to rational maps on the space of sequences $( a_{-1}, a_0, \dots )$.
Then one can regard the computation of successive Laurent series in an admissible pattern as a kind of multivalued dynamical system governed by words in the affine Weyl group $W(A_1^{(1)})$.
For example, the singularity pattern \eqref{singpatternsimple} corresponds to iterating from $a_{-1}=0$ according to the word $T r T$ as follows:
\begin{equation}
\operatorname{rg} \xrightarrow{~T~} \infty^{1}_{-1} \xrightarrow{~r~} \infty^{1}_{+1} \xrightarrow{~T~} \operatorname{rg}. 
\end{equation}
Note that the word $T r T$ is equivalent to the identity element in $W(A_1^{(1)})$. It turns out that any such word gives an admissible sequence that is confining in a similar way. 
For example it can be verified by direct calculation that the delay Painlev\'e-I equation admits the singularity pattern
\begin{equation} \label{singpattern2}
(\operatorname{rg}, \infty^{1}_{-1}, \infty^{1}_{-2},  \infty^{1}_{+2}, \infty^1_{+1}, \operatorname{rg}),
\end{equation}
which corresponds to the word $T^2 r T^2$:
\begin{equation}
\operatorname{rg} \xrightarrow{~T~} \infty^{1}_{-1} \xrightarrow{~T~} \infty^{1}_{-2} \xrightarrow{~r~} \infty^{1}_{+2}\xrightarrow{~T~}  \infty^{1}_{+1} \xrightarrow{~T~} \operatorname{rg}. 
\end{equation}
The next result follows from the definition of the rational maps $r$ and $T$ on formal Laurent series as above.
\begin{theorem}
For every word $w_m \cdots w_2 w_1$ with $w_i \in \left\{r, T \right\}$ which is equivalent to the identity in the affine Weyl group $W(A_1^{(1)})$, there exists an admissible sequence for the delay Painlev\'e-I  equation beginning and ending in Taylor series, obtained by successively applying $w_1$, $w_2$, $\dots$, $w_m$ to a Taylor series. 
\end{theorem}

This shows that the delay Painlev\'e-I  equation admits infinitely many confining singularity patterns, which has also been shown to be the case in other examples of delay-differential analogues of Painlev\'e equations \cite{delaySC}.
While affine Weyl groups and their extensions play a prominent role in the theory of differential and discrete Painlev\'e equations, to our knowledge this is their first appearance in the context of delay-differential Painlev\'e equations. 
It is not yet clear what the significance of this structure is, but we believe it warrants further investigation.

In Section \ref{sec:concludingremarks} we will see how the delay Painlev\'e-I equation arises from a contiguity relation satisfied by solutions of Painlev\'e-IV which are related by a B\"acklund transformation corresponding to a discrete Painlev\'e equation known as d$\pain{I}$.
Both Painlev\'e-IV and d$\pain{I}$ have surface type $E_6^{(1)}$ in the Sakai classification \cite{SAKAI2001}, for which the associated symmetry type in the generic case is $A_2^{(1)}$ (which is quite natural from the dressing chain point of view \cite{VS,Adler}).
Therefore the appearance of $A_1^{(1)}$ in the singularity structure of the delay Painlev\'e-I equation here may at first glance be surprising. 

However the d$\pain{I}$ equation is in fact associated with a non-translation element of infinite order in the symmetry group $\widetilde{W}(A_2^{(1)})$ of a generic $E_6^{(1)}$-surface, and has less than the full parameter freedom for its surface type,
a phenomenon known as projective reduction \cite{projectivereduction}. 
Such discrete Painlev\'e equations can be understood in terms of the normalizer theory of Coxeter groups \cite{yangnormalizers}, and we expect that d$\pain{I}$ is associated to a translation element of a subgroup of type $A_1^{(1)}$, explaining the appearance of $A_1^{(1)}$ in the singularity structure of the delay Painlev\'e-I equation.

\section{Link with the calculation of Masur-Veech volumes} \label{sec:masurveech}

Consider the moduli space $\mathcal{M}_{g,n}$ of algebraic curves $\mathcal{C}$ of genus $g$ with $n$ marked points. Its cotangent bundle can be identified with
the moduli space $\mathcal{Q}_{g,n}$ of pairs $(\mathcal{C}, q)$, where $q$ is a quadratic differential on $\mathcal{C}$ with simple poles at the marked points, and thus can be supplied with the canonical symplectic structure, and hence with corresponding canonical volume form. A quadratic differential $q$ defines a flat metric $|q|$ (with conical singularities) on the curve $\mathcal C$ with a finite area $Area(\mathcal C,q)$.
The \emph{Masur-Veech volume} $\operatorname{Vol} \mathcal{Q}_{g,n}$ introduced in 
\cite{masur, veech} is the volume of the subset of  $\mathcal{Q}_{g,n}$ with $Area(\mathcal C,q)\leq 1/2$.

These volumes have been extensively studied in recent years (see \cite{ADGZZ} and references therein). In particular, Chen, M\"oller and Sauvaget \cite{CMS} found a formula for $\operatorname{Vol} \mathcal{Q}_{g,n}$ as Hodge integrals over Deligne-Mumford compactifications of $\mathcal{M}_{g,n}$, and Kazarian \cite{KAZARIAN1} derived a recursion relation for the Masur-Veech volumes using a connection to the KP hierarchy from \cite{KAZARIAN0}.

Recently Yang, Zagier and Zhang \cite{YZZ} used the formula of Chen, M\"oller and Sauvaget and the results of Buryak \cite{Buryak} to derive another recursive procedure for calculating $\operatorname{Vol} \mathcal{Q}_{g,n}$.
In particular, they showed that the {\it Masur-Veech free energy} defined as the generating function
\begin{equation} \label{MVgenfunction}
\mathcal{H}(x,\epsilon) = \sum_{g,n \geq 0} \epsilon^{2g-2} \frac{x^n}{n!} a_{g,n},
\end{equation}
of the rational numbers $a_{g,n}$ related to the Masur-Veech volumes by
\begin{equation} \label{MVa}
\operatorname{Vol} \mathcal{Q}_{g,n}=2^{2g+1}\frac{\pi^{6g-6+2n}(4g-4+n)!}{(6g-7+2n)!} a_{g,n},
\end{equation}
satisfies the equation
\begin{equation} \label{Hdelayeq}
\left[ \partial_x \left( \mathcal{H}_+ - \mathcal{H}_- \right) \right]^2  + \partial_x^2 \left( \mathcal{H}_+ + \mathcal{H}_- \right) = \frac{2 x }{\epsilon^2},
\end{equation}
where $\mathcal{H}_{\pm}(x,\epsilon) =  \mathcal{H}\left(x \pm \frac{ i \epsilon}{2}, \epsilon \right)$. 

Our first observation is the following relation between the YZZ equation \eqref{Hdelayeq} and the delay Painlev\'e-I  equation.
Define the function
\begin{equation} \label{PhidefMV}
F(x,\epsilon): = \partial_x \left(\mathcal{H}_+ - \mathcal{H}_-\right).
\end{equation}

\begin{proposition} \label{HeqntodelayP1}
Let $\mathcal{H}(x,\epsilon)$ be a solution of equation \eqref{Hdelayeq},
then the function
\begin{equation} \label{MVfdefinition}
f_{MV}(z): = -i \epsilon F\left( i \epsilon\left( z - \tfrac{1}{2} \right)  , \epsilon\right),
\end{equation}
 is a formal entire solution of the delay Painlev\'e-I equation \eqref{scaledperturbation} with parameter $\beta=i \epsilon$.
\end{proposition}
The proof is straightforward.
The question is then how to characterise this very special solution of the delay Painlev\'e-I equation.
In terms of the parameters $s_k$ the answer turns out to be remarkably simple.


\begin{theorem} \label{th:masurveechpolynomials}
The function $f_{MV}(z)$ constructed from the Masur-Veech free energy by formula (\ref{MVfdefinition})
coincides with the special formal solution of the delay Painlev\'e-I equation \eqref{delayP1}, satisfying \eqref{new1} with parameters $s_\ell=0,\,\, \ell\in \mathbb N$.
More explicitly,
\begin{equation} \label{SPEC}
f_{MV}(z)=  \sum_{\ell=1}^{\infty} Q_{\ell}(z) \beta^{\ell},
\end{equation}
where $Q_{\ell}(z)=\Phi_\ell(z,\tau^*)$ are the Bernoulli-Catalan polynomials \eqref{Qnrel}.

In particular, the corresponding special parameters $\tau^*$ in \eqref{special} are related to the Masur-Veech volumes according to 
\begin{equation} \label{SPEC1}
\tau^*_{\ell}= \sum_{0\leq g\leq \frac{\ell}{2}} (-1)^g \frac{2a_{g,l-2g+2}}{(l-2g+1)!},
\end{equation} 
for $\ell$ even, and vanish for $\ell$ odd.
\end{theorem}

\begin{proof}
With $\beta=i\varepsilon$, by definition 
$$f_{MV}(z)=-\beta \partial_x (e^{\frac{1}{2}\beta \partial_x}-e^{-\frac{1}{2}\beta \partial_x}) \mathcal H(x,\varepsilon), \quad x=\beta(z-\tfrac{1}{2}).$$
The YZZ equation (\ref{Hdelayeq}) for $u(x,\varepsilon):=f_{MV}(z), \, x=\beta(z-\frac{1}{2})$ takes the form
$$
\frac{u^2}{\beta^2}-\frac{e^{\frac{1}{2}\beta \partial_x}+e^{-\frac{1}{2}\beta \partial_x}}{e^{\frac{1}{2}\beta \partial_x}-e^{-\frac{1}{2}\beta \partial_x}}\frac{\partial_x u}{\beta}=-\frac{2 x }{\beta^2},
$$
or, after multiplication by $\beta^2$:
\begin{equation} \label{SPEC3}
u^2-\frac{e^{\beta \partial_x}+1} {e^{\beta \partial_x}-1}\beta \partial_x u=-2x.
\end{equation}
Comparing this with the equation (\ref{new1})
$$
u^2=\frac{e^{\beta D}+1} {e^{\beta D}-1} \beta D u  - 2\xi -\sum_{k=1}^\infty s_k \beta^k, \quad \xi=\beta(z-\tfrac{1}{2}), \quad D=\partial_\xi,
$$
we see that they coincide when all the constants of integration $s_{k}=0.$
This proves relation \eqref{SPEC}.

Comparing now  \eqref{SPEC} with the definition of $f_{MV}(z)$ we have the following expression of the Bernoulli-Catalan polynomials in terms of Masur-Veech volumes
\begin{equation} \label{qldefinition}
Q_{\ell}(z)  = \sum_{0\leq g\leq \lfloor \frac{\ell}{2}\rfloor}(-1)^g a_{g,l-2g+2}\frac{\Delta z^{l-2g+1}}{(l-2g+1)!},
\end{equation} 
where we defined $$
\Delta z^k:=z^k-(z-1)^k, \quad k \in \mathbb N.
$$
To prove relation \eqref{SPEC1} recall that $\tau_\ell=\Phi_\ell(0)+\Phi_\ell(1)=Q_\ell(0)+Q_\ell(1)$. 
Since $$\Delta z^k(0)+\Delta z^k(1)=1+(-1)^k$$ from (\ref{qldefinition}) we see that the corresponding parameters $\tau_\ell=\tau_\ell^*$ are zero for odd $\ell$ and are given by  (\ref{SPEC1}) for even $\ell.$
\end{proof}

Let $A_k^{(n)}:=a_k^{(n)}(\tau^*)$ as before be the coefficients of the Bernoulli-Catalan polynomials
$$
Q_n(z)=\sum_{k=0}^nA_k^{(n)} z^{k}.
$$

\begin{corollary}
The Masur-Veech numbers can be expressed via the coefficients $A_k^{(n)}$ by the formula
\begin{equation} \label{MVBer}
a_{g,n}=(-1)^g(n-2)!\sum_{k=0}^{2g}(-1)^k  \binom{n+2g-k-2 }{2g-k}B_{2g-k}A_{n+2g-2-k}^{(n+2g-2)},
\end{equation}
where $B_m$ are the Bernoulli numbers.
\end{corollary}

\begin{proof}
Note that since $\tau^*_{\ell}=0$ for odd $\ell$, the corresponding polynomials have the symmetry property $Q_{\ell}(1-z) = (-1)^{\ell} Q_{\ell}(z)$, implying that
\begin{equation} \label{plus1}
Q_{\ell}(z+1) = (-1)^{\ell} Q_{\ell}(-z)=\sum_{k=0}^\ell (-1)^{\ell-k} A_k^{(\ell)} z^{k}.
\end{equation}
On the other hand, from (\ref{qldefinition}) and (\ref{SPEC}) we have
\begin{equation} \label{plus2}
Q_{\ell}(z+1)= \sum_{0\leq g\leq \lfloor \frac{\ell}{2}\rfloor}(-1)^g a_{g,l-2g+2}\frac{(T-1)z^{l-2g+1}}{(l-2g+1)!},
\end{equation}
where $T$ as before is the shift $Tf(z) = f(z+1)$. 
Recall now that the Bernoulli polynomials 
$$
B_n(z):=\sum_{k=0}^n \binom{n}{k}B_k z^{n-k},
$$
have the property
$
(T-1)B_n(z)=n z^{n-1},
$
so 
\begin{equation} \label{plus3}
z^k=\frac{1}{k+1}(T-1)B_{k+1}(z)=\frac{1}{k+1}\sum_{j=0}^{k} \binom{k+1}{j}B_j (T-1)z^{k+1-j}.
\end{equation}

Combining the relations (\ref{plus1}), (\ref{plus2}), (\ref{plus3}), we have the claim.
\end{proof}

We can use our results to derive the formula
\begin{equation}
\operatorname{Vol} \mathcal{Q}_{0,n} = \frac{\pi^{2n-6}}{2^{n-5}}, \quad n \geq 3,
\end{equation}
which was conjectured by Kontsevich and first proved by Athreya, Eskin and Zorich in \cite{AEZ}.

Indeed, for $g=0$ the formula (\ref{MVBer}) gives
$$
a_{0,n}=(n-2)!A_{n-2}^{(n-2)}=(2n-7)!!=\frac{(2n-7)!}{2^{n-4}(n-4)!},
$$
where we used Proposition \ref{leadq} saying that
$
A_{\ell}^{(\ell)} =\frac{ (2\ell-3)!!}{\ell!}.
$
Now using (\ref{MVa}) we have
\begin{equation}
\operatorname{Vol} \mathcal{Q}_{0,n} =2 \frac{\pi^{2n-6}(n-4)!}{(2n-7)!}a_{0,n}= 2 \frac{\pi^{2n-6}(n-4)!}{(2n-7)!}\frac{(2n-7)!}{2^{n-4}(n-4)!}=\frac{\pi^{2n-6}}{2^{n-5}},
\end{equation}
which is the Kontsevich formula.

Similarly, for $g=1$ formula (\ref{MVBer}) gives
$$
a_{1,n}=-(n-2)!\left[\binom{n}{2}B_2A_n^{(n)}-\binom{n-1}{1}B_1A_{n-1}^{(n)}+\binom{n-2}{0}B_0A_{n-2}^{(n)}\right].
$$
From Proposition \ref{leadq} we have
$$
A_{\ell}^{(\ell)} =\frac{ (2\ell-3)!!}{\ell!}, \quad A_{\ell-1}^{(\ell)}=-\frac{\ell}{2}A_{\ell}^{(\ell)},\quad A_{\ell-2}^{(\ell)} =\frac{\ell(\ell-1)}{8}A_{\ell}^{(\ell)}-\frac{2^{\ell-2}(\ell-1)}{12}.
$$
Using this we can find $a_{1,\ell}$ as
$$
a_{1,\ell}=\frac{(2\ell -3)!!}{24}+\frac{2^{\ell-2}(\ell-1)!}{12},
$$
and hence the formula for the Masur-Veech volume
\begin{equation}\label{MV2}
\operatorname{Vol}  Q_{1,n}=\frac{\pi^{2n}}{3}\left[\frac{n!}{(2n-1)!!}+\frac{n}{2^{n-1}(2n-1)}\right],
\end{equation}
which was first conjectured by Anderson et al \cite{ABCDGLW} and proved by Chen et al \cite{CMS} (see also \cite{YZZ}).

We should make clear that our derivation is based on the results of Yang, Zagier and Zhang \cite{YZZ}, and thus cannot be considered as new independent proof of these formulas. We are using these calculations mainly to justify the further study of our polynomials $\Phi_n$. One of the most important questions is to provide an alternative description of the special values of parameters $\tau_\ell^*$.

Note that in \cite{YZZ} it was also shown that the Masur-Veech free energy satisfies in addition to \eqref{Hdelayeq} a second equation 
\begin{equation} \label{MVsecondeq}
\left( \epsilon \partial_{\epsilon} + \frac{1}{2} x \partial_x - \frac{\epsilon^2}{24} \partial_x^3 \right) \left( \mathcal{H}_+ -\mathcal{H_-} \right) + \frac{\epsilon^2}{12} \left[ \partial_x \left(\mathcal{H}_+ - \mathcal{H}_- \right) \right]^3 = 0,
\end{equation}
as well as a certain boundary condition (see \cite{YZZ} for details). 
We observe that this second equation \eqref{MVsecondeq} is nothing but a special case of the second modified KdV equation \cite{nakamurahirota}. Indeed, letting $h(x,\epsilon)\defeq \mathcal{H}_+ - \mathcal{H}_-$, then introducing $\xi = x \epsilon^{-1/2}$, $t = \epsilon^{1/2}$, $v(\xi,t) = h(x,\epsilon)$,
we obtain 
\begin{equation}
v_t = \frac{1}{12} \left(v_{\xi \xi \xi} - 2 v_{\xi}^3 \right).
\end{equation}
The equation \eqref{MVsecondeq} implies that the family of solutions $f = f_{MV}(z;\beta)$ of delay Painlev\'e-I  parametrised by $\beta$ as defined by \eqref{MVfdefinition} satisfy the equation
\begin{equation}
 2 f_{\beta} \beta^2 - \left(f + \left(z-\tfrac{1}{2}\right)  f_z \right)\beta- \frac{1}{12} f_{zzz} +  \frac{1}{2}f_z f^2 = 0.
\end{equation}
In terms of the polynomials $Q_{\ell}(z)$ in the series $f_{MV}(z;\beta) = \sum_{\ell=1}^{\infty} Q_{\ell}(z) \beta^{\ell}$ this becomes the third order recurrence relation
\begin{equation}
Q_1''' = 0, \quad 
\frac{1}{6} Q_{\ell}''' + (4\ell-6) Q_{\ell-1} - 2(z-\tfrac{1}{2}) Q_{\ell-1}' - \sum_{k=1}^{\ell-1} Q_k' \sum_{j=1}^{\ell-k-1}Q_j Q_{\ell-j-k} = 0, \quad \ell \geq 2. 
\end{equation}


\section{Concluding remarks} \label{sec:concludingremarks}

The Painlev\'e-$\mathrm{I}$ equation plays a special role among all Painlev\'e equations, being the simplest and the only one without parameters. Its solutions are remarkable meromorphic transcendental functions, which were studied for more than a hundred years starting from Boutroux \cite{Boutroux}, and appear in many important areas, including quantum gravity.

A discrete version of Painlev\'e-$\mathrm{I}$ equation first appeared in the theory of orthogonal polynomials in the work of Shohat \cite{shohat}. Its general form is  
\begin{equation} \label{discreteP1}
x_{n+1} + x_n + x_{n-1}  = \frac{a n + b}{x_n} + c, 
\end{equation}
where $a,b,c$ are parameters such that $a\neq0$. Shohat proved that the coefficients $x_n=a_n^2$ in the three-term recurrence relation 
$$
P_{n+1}-xP_n+a_n^2 P_{n-1}=0,
$$
for the monic orthogonal polynomials $P_n=P_n(x;t)$ with the exponential weight $w(x;t)=e^{tx^2-x^4}$ satisfy the nonlinear recurrence \eqref{discreteP1} with $a=\frac{1}{4}, \, b=0, \, c=\frac{t}{2}$.


Magnus \cite{Magnus} had shown that
the d$\pain{I}$ equation \eqref{discreteP1} provides a B\"acklund transformation for the fourth Painlev\'e equation in the form
\begin{equation} \label{P4scalar}
\frac{d^2x}{dt^2} = \frac{(\frac{dx}{dt})^2}{2x} + \frac{3}{2} x^3 + 4 t x^2 + 2 (t^2 - \alpha)x + \frac{\beta}{x}.
\end{equation}
Namely, let $x=x_n(t)$ be a solution of \eqref{P4scalar} with parameters 
$$\alpha = \frac{1}{2}(n-b), \quad \beta = - \frac{1}{2}(n-b)^2,$$ 
then 
\begin{equation}
x_{n+1} = \frac{b - n - 2 t x_n - x_n^2 + \frac{dx_n}{dt}}{2 x_n},
\end{equation}
gives a solution of \eqref{P4scalar} with the same parameters but with $n\to n+1$. 
There is also an inverse transformation corresponding to the parameter shift $n\to n-1$, and the sequence $x_n, n\in \mathbb{Z}$ of such solutions satisfy the d$\pain{I}$ equation \eqref{discreteP1} with $a = -1$, $c=-2t$.

Further, solutions $x_n$ and $x_{n+1}$ related by this B\"acklund transformation satisfy a kind of contiguity relation, which takes the form of the differential-difference equation
\begin{equation} \label{contiguity}
\frac{dx_{n+1}}{dt} + \frac{dx_{n}}{dt} = x_{n}^2 - x_{n+1}^2 + 2 t (x_n-x_{n+1}) - 1.
\end{equation}
Combining a change of variables with a delay reduction by introducing 
\begin{equation}
f(z) = - i \sqrt{\alpha} \,\left(x_n(t) + t \right), \quad z = \frac{i}{\sqrt{\alpha}} t + n h,
\end{equation}
the equation \eqref{contiguity} becomes
\begin{equation}
f'(z+h) + f'(z) = f(z+h)^2 - f(z)^2 + \alpha,
\end{equation}
which is the delay Painlev\'e-I equation \eqref{delayP1}.
Note that one can consider here a more general closure of the dressing chain with
$$
f_{j+1}(z)=q f_j(q z+h),
$$
but for $q \neq 1$ this can be reduced to the closure $f_{j+1}(z)=qf_j(qz)$, studied in detail by Shabat \cite{Shabat}, Adler \cite{Adler2} and Liu \cite{Liu}.


The delay Painlev\'e-$\mathrm{I}$ equation, which we considered here, seems to be the most natural delay-differential version of the Painlev\'e-$\mathrm{I}$ equation.
It is related to the ILW equation in the same way as the Painlev\'e-$\mathrm{I}$ equation is related to the KdV equation.

Note that the ILW equation with a small parameter $\varepsilon$
\begin{equation}
u_t = 2 u u_x - i\varepsilon\frac{T+1}{T-1} u_{x x}, \quad T=e^{i\varepsilon\frac{\partial}{\partial x}},
\end{equation}
can be re-written as
\begin{equation}
\label{newILW}
\frac{\partial u}{\partial t} = 2 u \frac{\partial u}{\partial x} - 2\sum_{g=1}^\infty \varepsilon^{2g} \frac{(-1)^g B_{2g}}{(2g)!} \frac{\partial ^{2g+1} u}{\partial x^{2g+1}},
\end{equation}
where $B_n$ are Bernoulli numbers. 
In this form the ILW equation appeared in the work of Buryak \cite{Buryak}, who proved that the Hodge integrals satisfy the ILW hierarchy, which was used by Yang, Zagier and Zhang \cite{YZZ} to derive the key relation \eqref{Hdelayeq}. 

This makes the appearance of Bernoulli-like polynomials in relation with the delay Painlev\'e-$\mathrm{I}$ equation and Masur-Veech volumes not that surprising, but, unlike in the Bernoulli case, the recurrence relation for the new polynomials $\Phi_n$ is nonlinear of Catalan type, which makes their study much more difficult. 
We should mention that there are polynomial versions of Catalan numbers known, in particular, from the theory of $q,t$-Catalan numbers (see e.g. \cite{garsiahaiman} and \cite{Stanley99}, Ch. 6), but they seem to be of a different type. 
In any case, the question about possible combinatorial interpretation of the coefficients of the new polynomials is very natural. 

 Another interesting parallel is between the polynomials $\Phi_n(z;\tau)$ and the Burchnall-Chaundy (also known as Adler-Moser) polynomials $P_n(z;\tau)$, which also satisfy a nonlinear recurrence relation and depend on parameters related to the times $t_k$ in the KdV hierarchy by a triangular polynomial transformation (see \cite{adlermoser}, \cite{ralphsasha}).

Finally, we would like to explore if our approach can help to prove some conjectures from \cite{ADGZZ, YZZ} about asymptotics of Masur-Veech volumes $\operatorname{Vol} \mathcal{Q}_{g,n}$ for large $g.$

%
%

\subsection*{Acknowledgements} 
We are grateful to Vsevolod Adler for very helpful discussions and comments on the earlier version of this paper.

This work was started while AS was supported by a London Mathematical Society Early Career Fellowship, and AS gratefully acknowledges the support of the London Mathematical Society, in particular, to visit Loughborough University in September 2021.
AS was supported by a Japan Society for the Promotion of Science (JSPS) Postdoctoral Fellowship for Research in Japan and also acknowledges the support of JSPS KAKENHI Grant Numbers 21F21775 and 22KF0073. 
Part of the preparation of this manuscript was done while AS was visiting the Okinawa Institute of Science and Technology (OIST) for a week in October 2023, through the Theoretical Sciences Visiting Program (TSVP) as part of the TSVP Thematic Program ``Exact Asymptotics: From Fluid Dynamics to Quantum Geometry".

\bibliographystyle{amsalpha}

\end{document}